\documentclass{article}

\usepackage[english]{babel}
\usepackage[latin1]{inputenc}
\usepackage{amsmath}
\usepackage{amsfonts}
\usepackage{amssymb}
\usepackage{enumerate}
\usepackage{stmaryrd}
\usepackage{alltt}
\usepackage{bbm}
\usepackage{graphicx}
\usepackage[all]{xy}
\usepackage{paralist}
\usepackage{url}
\usepackage{mathrsfs}

\usepackage{bussproofs}
\usepackage[latin1]{inputenc}

\EnableBpAbbreviations

\def\halfthinspace{\relax\ifmmode\mskip.5\thinmuskip\relax\else\kern.8888em\fi}\let \hts=\halfthinspace

\def\({\left (}
\def\){\right )}

\def\ajoin{{\hts + \hts}}
\def\aunit{1}
\def\ameet{{\hts \cdot \hts}}
\def\azero{0}
\def\acompl#1{\overline{#1}}

\def\acompo{{\hts ; \hts}}
\def\aid{1^\prime}
\def\aconv#1{\setbox13\hbox{$#1$}\ifdim\wd13<12pt\breve{#1}\else{\(#1\)}\breve{\ }\fi}

\def\aclosure#1{{#1}^{*}}

\def\pjoin{\cup}
\def\punit#1{#1}
\def\pmeet{\cap}
\def\pzero{\emptyset}

\def\pcomple{\mbox{$^{\mbox{--}}$}}
\def\pcompo{{\hts \circ \hts}}
\def\pid{Id}
\def\pconv#1{\setbox13\hbox{$#1$}\ifdim\wd13<10pt\stackrel{\smile}{#1}\else{\(#1\)}^{\smile}\fi}
\def\pconve{^{\smile}}

\def\pclosure#1{{#1}^{\underline{*}}}

\def\setof#1#2{\left \{\, #1 : #2 \,\right \}}
\def\set#1{\left \{\, #1 \,\right \}}
\def\pair#1#2{\left \langle #1, #2 \right \rangle}

\def\nat{\Longrightarrow}
\def \op{{\sf op}}
\def \<{\left\langle}
\def \>{\right\rangle}
\def \({\left(}
\def \){\right)}

\def \setof#1#2{\setbox1\hbox{$#1$}
                \setbox2\hbox{$#2$}
                \ifdim \ht1 > \ht2
                   \left \{ \left . \, #1 \, \right \vert \, #2 \, \right \}
                \else
                   \left \{ \, #1 \, \left \vert \, #2 \, \right . \right \} 
                \fi}
\def \set#1{\left \{\, #1 \,\right \}}
\def \pair#1#2{\left \langle #1, #2 \right \rangle}

\newcommand{\X}{\mathsf{X}}        
\newcommand{\U}{\mathsf{U}}

\newcommand{\Ex}{\mathsf{EX}}

\newcommand{\EG}{\mathsf{EG}}

\newcommand{\eee}{\mathsf{E}}

\newcommand{\NAT}{ I\hspace{-0.12cm}N } 

\newtheorem{definition}{Definition}
\newtheorem{example}{Example}

\newtheorem{proposition}{Proposition}
\newtheorem{theorem}{Theorem}
\newtheorem{fact}{Fact}

\newtheorem{corollary}{Corollary}
\newtheorem{lemma}{Lemma}
\def \qed{\ifmmode\rule{5pt}{5pt}\else{\nobreak\hfil\penalty50\hskip1em\null\nobreak\hfil\rule{5pt}{5pt}\parfillskip=0pt\finalhyphendemerits=0\endgraf}\fi}
\newenvironment{proof}{\emph{\textbf{Proof:}}\rm}{\qed \par \medskip}

\usepackage{xspace}
\usepackage{xcolor}
\usepackage{framed,color}
\usepackage{todonotes}
\definecolor{LavanderPosta}{HTML}{ebd9fc}

\title{A proof theoretic basis for relational semantics\thanks{Carlos G.\ Lopez Pombo's research is supported by Universidad de Buenos Aires through grant UBACyT 20020170100544BA and Agencia Nacional de Promoci\'on de la Investigaci\'on, el Desarrollo Tecnol\'ogico y la Innovaci\'on Cient\'{\i}fica through grant PICT-2019-2019-01793. Thomas S.E. Maibaum's research is supported by the Canada Research Chairs Program through grant 950-225524, General Motors of Canada Ltd through contract PO \# 4300478369, Natural Sciences and Engineering Research Council through grant CRDPJ-515486-2017, and Ontario Centres of Excellence through grant OCE \# 30040; Carlos G. Lopez Pombo's research was also supported by the following grants to Thomas S. E. Maibaum: Ontario Ministry of Economic Development, Job Creation and Trade, Ontario Research Fund - Research Excellence Program ORF-RE 03-045, Natural Sciences and Engineering Research Council, Automotive Partnership Canada GRFN APCPJ386797-09, Natural Sciences and Engineering Research Council of Canada GRFN STPGP430575-12, Natural Sciences and Engineering Research Council GRFN RGPIN26775-11, Ontario Ministry of Research and Innovation, Ontario Research Fund - Research Excellence Program ORF-RE 05-044.}}

\author{
Carlos~G.~Lopez~Pombo\\
\small{Department of Computing, FCEyN, Universidad de Buenos Aires}\\
\small{Consejo Nacional de Investigaciones Cient{\'{\i}}ficas y Tecnol{\'o}gicas (CONICET)}\\
\small{\url{clpombo@dc.uba.ar}}
\and 
Thomas~S.E.~Maibaum\\
\small{Department of Computing and Software, McMaster University}\\
\small{\url{tom@maibaum.org}}
}

\date{\today}

\begin{document}

\maketitle

\begin{abstract}
Logic has proved essential for formally modeling software based systems. Such formal descriptions, frequently called specifications, have served not only as requirements documentation and formalisation, but also for providing the mathematical foundations for their analysis and the development of automated reasoning tools.\\
Logic is usually studied in terms of its two inherent aspects: syntax and semantics. The relevance of the latter resides in the fact that producing logical descriptions of real-world phenomena, requires people to agree on how such descriptions are to be interpreted and understood by human beings, so that systems can be built with confidence in accordance with their specification. On the more practical side, the metalogical relation between syntax and semantics, determines important aspects of the conclusions one can draw from the application of certain analysis techniques, like model checking.\\
Abstract model theory (i.e., the mathematical perspective on semantics of logical languages) is of little practical value to software engineering endeavours. From our point of view, values (those that can be assigned to constants and variables) should not be just points in a platonic domain of interpretation, but elements that can be named by means of terms over the signature of the specification. In a nutshell, we are not interested in properties that require any semantic information not representable using the available syntax.\\
In this paper we present a framework supporting the proof theoretical formalisation of classes of relational models for behavioural logical languages, whose domains of discourse are guaranteed to be formed exclusively by nameable values.
\end{abstract}

\section{Introduction}
\label{introduction}

Logic has proved essential as a formal tool for describing, and then reasoning about, different aspects of the world we perceive. The formal modelling of software artefacts is a widely known, and accepted, example of its usefulness. Many formal languages have been devised in order to reflect different aspects of the behaviour of software systems; among many examples, one can mention linear time temporal logics, both propositional \cite{pnueli:ieee-focs77,pnueli:tcs-13_1} and first-order \cite{manna95}, branching time temporal logics \cite{benari:acm-sigplan-sigact81}, $CTL^*$ \cite{emerson:jacm-33_1} as a compromise between linear-time and branching-time logics, the many versions of dynamic logics \cite{fischer:stoc77,harel00} and dynamic linear temporal logic \cite{henriksen:apal-96_1_3} in order to try to capture the links between dynamic and linear temporal logics, higher-order logics \cite{vanbenthem:hlfcs83} for capturing several notions like higher-order functions, equational logic \cite{henkin:amm-84_8} for capturing abstract data types \cite{ehrig:cj-35_5}, etc. 

Software analysis is an area in software engineering (and computer science) concerned with the application of automatic, and semi-automatic, techniques aimed at proving the (relative) absence of (certain) errors, or the compliance with specific quality standards, resulting from the satisfaction of certain desired properties of their formal specifications. Many of the aforementioned logical languages have associated tools enabling their use, not only as specification languages for formalising the requirements or the designs of a system, but also for providing some sort of mechanised analysis.

As usual, formalising a software based system by resorting to a logical language requires people involved in the development process to agree on how such descriptions are to be interpreted and understood, so that a system can be built in accordance with its specification and, therefore, have the expected behaviour. Semantics plays a central role in this endeavour as it provides a way of substituting the perhaps drier and more esoteric forms provided by syntactic descriptions, by the more intuitive modes of understanding appealing to some naive form of set theory. Therefore, the metalogical relation between these two inherent aspects of logical languages, syntax and semantics, determines important aspects of the conclusions one can draw from the application of their associated analysis techniques, like model-checking, that could be implemented over a representation of models for which the reasoning is not complete.

Abstract model theory \cite{chang++:barwise90,goguen:jacm-39_1} (i.e., the mathematical perspective on semantics of logical languages, generally concerned with the understanding of a class of mathematical structures) is of little practical value to software engineering endeavours. The logicians' perspective on semantics generally relies on descriptions given in naive set theory and, what is more troubling, the unquestioned hypothesis that the intuition about such descriptions is shared by the whole of the community; for example, the K4 system is interpreted over Kripke frames whose accessibility relation is transitive. These mathematical, semiformal tools might be considered enough for agreeing on the validity of a certain mathematical property, but, in the case of mission critical systems like those running on unmanned autonomous vehicles, they cannot be adopted for the analysis and verification stages of the certifications of safety critical properties. From our point of view, when semantics is to be used as a reasoning tool, values (those that can be assigned to constants and variables), should not be just points in a platonic domain of interpretation, but elements that can be named by means of terms over the signature of the specification. In a nutshell, we are not interested in properties that require any semantic information not representable using the available syntax, but only of those whose truth status is determined by concrete elements whose existence is witnessed by whether they can be constructed or not through the invocation of the available functions.

Many different classes of concrete models can be devised, each of which might be useful in its own use context. A traditional example from the literature is the definition of set-based classes of algebras, usually referred to as proper, in some sort of formal version of set theory; a different approach can be found in Rabe's work \cite{rabe:jlc-27_6} in which the author proposes defining models as morphisms to maximally consistent theories capable of determining the truth/falsity of any formula of the logical language. From our point of view, the strength of Rabe's approach is twofold: 
\begin{inparaenum}[1)]
\item underpinning the existence of well-defined composition mechanisms for logical specifications and combination mechanisms for logical languages, and
\item its generality, resulting from the disregard of any interpretation of logical operators over elements external to the logical language's syntax.
\end{inparaenum}
One can locate its weakness, if any, in the lack of a more intuitive view of the logical structure of individual models, and the properties of classes of such models, as a consequence of focusing on the syntactic aspects of the language, thus debilitating the idea that semantics plays an important role in the possibility of building agreements on the meaning of a logical language.

The direction we pursue in this paper shares its motivation with the one presented by Schlingloff and Heinle in \cite{schlingloff:brink97} (i.e., that in the context of the verification of safety critical systems, the model theory of behavioural specifications has to be formally defined). Instead of focussing on the study of modal logics from a relational algebraic point of view, we focus on providing a general and versatile, proof theoretically supported, framework for defining classes of relational models. These classes of models are widely used to provide semantics to a variety of modal \cite{kripke:apf-16}, hybrid \cite{areces:phdthesis} and deontic logics \cite{aqvist:hpl01,vonwright:mns-60_237}, ubiquitous in computer science and software engineering.

The framework we propose has two desirable properties. First, it is capable of capturing a wide range of logics (many of which are shown as examples in Sec.~\ref{examples}), making the effort of the formalisation worthwhile and, second, the semantics resulting from the use of the framework has to be widely understandable so as to enable the process of agreement about the intended meaning of the syntactic descriptions in the engineering context. (This is one of the major triumphs of the conventional naive set theory based approaches to semantics.)


As we mentioned in the preceding paragraphs, faithful to the standpoint of computer science and, more specifically, to that of software engineering, our proposal focuses on building a framework for describing classes of relational models, whose interpretation of both rigid and flexible symbols is done over concrete values (i.e., they can be denoted by terms).

The semantics of many specification languages used in the description of software artefacts, among which we can find many modal and hybrid logics, is defined over relational models \cite{kripke:zmlgm-9_56,kripke:ttm65}. In general, it is given in terms of the following common elements:
\begin{inparaenum}[1)] 
\item an \emph{interpretation} of a subset of symbols whose meaning is fixed for all states in which the system can be, usually referred to as \emph{rigid symbols};
\item a relational structure whose places are considered \emph{states} of the systems (also referred to as \emph{worlds}), serving the purpose of providing meaning for the symbols whose meaning can vary, usually referred to as \emph{flexible symbols};
\item some form of structuring of states (for example, infinite sequences of states in linear temporal logics, states in branching time temporal logics or dynamic logics, etc.) capable of interpreting the behaviour of the modal operators of the logical language; and
\item a notion of \emph{satisfaction} relating a structure (usually consisting of an interpretation, a relational structure and a specific structuring of states) to a formula depending on whether the latter is true when it is interpreted within the context of the former.
\end{inparaenum}
In general, the ordering of states must be coherent with the accessibility relation determined by the relational structure.

The contribution of this paper is centred on the definition of a unified framework for providing relational semantics to logical languages, within the field of Institutions \cite{goguen:cmwlp84,goguen:jacm-39_1}. Institutions have proved useful as a formal tool for:
\begin{inparaenum}[1)]
\item  providing a neat structuring of the relevant concepts of model theory, by resorting to tools coming from the field of category theory, and
\item  providing mechanisms for understanding concepts relevant to software engineering, such as modularity, parameterisation, heterogeneous description, etc.
\end{inparaenum}
On the one hand, we propose a formalisation of interpretations and sates as theories in equational logic \cite{henkin:amm-84_8}, extended with non-logical predicate symbols, providing the means for: 
\begin{inparaenum}[1)]
\item representing the atomic formulae (i.e., equality of terms like $t_1 = t_2$, provided that $t_1$ and $t_2$ are ground terms, and predicate symbols applied to an appropriate number of terms like $P (t_1, \ldots, t_n)$, provided that $t_1, \ldots, t_n$ are ground terms) that hold in a state, and consequently
\item representing values as ground terms.
\end{inparaenum}
On the other hand, we propose the formalisation of the relational structure of models as the models of a theory presentation in Tarski's \emph{Elementary Theory of (Binary) Relations} \cite{tarski:jsl-6_3}, extended with reflexive and transitive closure (and the necessary sentential elements required in order to obtain a complete axiomatisation of this new relational operator).

The paper is organised as follows: Sec.~\ref{institutions} presents the formal background required to understand the rest of the paper, Sec.~\ref{languages} present the basic logical languages over which the framework will be constructed,  Sec.~\ref{relmodels} present the definition of the framework, in Sec.~\ref{examples} we show the use of the framework for providing concrete semantics to several well-known modal logics and, finally in Sec.~\ref{conclusions} we draw some conclusions and outline further lines research.

\section{Institutions and General logics}
\label{institutions}
From now on we assume that the reader has a nodding acquaintance with category theory and is familiar with the basic definitions of the field. (See \cite{vanoosten:BRICS-LS-95-1} for a quick reference or \cite{maclane71,pierce91} for a more thorough presentation.)

The theory of institutions, initially presented by Goguen and Burstall in \cite{goguen:cmwlp84}, provides a formal and generic definition of what a logical system is, from an abstract model theoretical point of view. This work evolved in many directions: in \cite{meseguer:lc87}, Meseguer complemented the theory of institutions by providing a categorical characterisation for the notions of entailment system (also called $\pi$-institutions by Fiadeiro et. al. in \cite{fiadeiro:icdcwtfm93}) and the corresponding notion of proof calculi; in \cite{goguen:jacm-39_1,tarlecki:sadt-rtdts95} Goguen and Burstall, and Tarlecki, respectively, extensively investigated the ways in which institutions can be related, among which theoroidal co-morphisms have a distinguished role by providing a notion of semantics preserving representation of a logical system into another.

Let us review the definitions we will need throughout the present work.

An \emph{entailment system} is defined by identifying a family of \emph{syntactic} consequence relations. Each of the elements in this family is associated with a signature. These relations are required to satisfy reflexivity, monotonicity and transitivity.

\begin{definition}[Entailment system \cite{meseguer:lc87}] 
\label{entailment-system}
An \emph{entailment system} is a structure of the form $\< \mathsf{Sign}, \mathbf{Sen}, \{\vdash_{\Sigma}\}_{\Sigma \in |\mathsf{Sign}|} \>$ satisfying the following conditions:
\begin{itemize}[$-$]
\item $\mathsf{Sign}$ is a category of signatures,
\item $\mathbf{Sen}: \mathsf{Sign} \to \mathsf{Set}$ is a functor.
\item $\{\vdash_{\Sigma}\}_{\Sigma \in |\mathsf{Sign}|}$, where $\vdash_{\Sigma} \subseteq 2^{\mathbf{Sen}(\Sigma)} \times \mathbf{Sen}(\Sigma)$, is a family of binary relations such that for any $\Sigma, \Sigma' \in |\mathsf{Sign}|$, $\{\phi\} \cup \{\phi_i\}_{i \in  \mathcal{I}} \subseteq \mathbf{Sen}(\Sigma)$, $\Gamma, \Gamma' \subseteq \mathbf{Sen}(\Sigma)$, the following conditions are satisfied:
\begin{enumerate}
\item \label{cond-1} reflexivity: $\{\phi\} \vdash_\Sigma \phi$,
\item \label{cond-2} monotonicity: if $\Gamma \vdash_\Sigma \phi$ and $\Gamma \subseteq \Gamma'$, then $\Gamma' \vdash_\Sigma \phi$,
\item \label{cond-3} transitivity: if $\Gamma \vdash_\Sigma \phi_i$ for all $i \in \mathcal{I}$ and $\{\phi_i\}_{i \in \mathcal{I}} \vdash_\Sigma \phi$, then $\Gamma \vdash_\Sigma \phi$, and
\item \label{cond-4} $\vdash$-translation: if $\Gamma \vdash_\Sigma \phi$, then for all $\sigma: \Sigma \to \Sigma' \in ||\mathsf{Sign}||$, $\mathbf{Sen}(\sigma)(\Gamma) \vdash_{\Sigma'} \mathbf{Sen}(\sigma)(\phi)$. 
\end{enumerate}
\end{itemize}
\end{definition}

\begin{definition}[Theory \cite{meseguer:lc87}]
Let $\< \mathsf{Sign}, \mathbf{Sen}, \{\vdash_{\Sigma}\}_{\Sigma \in |\mathsf{Sign}|} \>$ be an entailment system and $\<\Sigma, \Gamma\> \in |{\sf Th}|$. We define the function $^\bullet: 2^{\mathbf{Sen}(\Sigma)} \to 2^{\mathbf{Sen}(\Sigma)}$ as $\Gamma^\bullet = \setof{\gamma}{\Gamma \vdash_\Sigma \gamma}$. This function is extended to elements of ${\sf Th}$, by defining it as follows: $\<\Sigma, \Gamma\>^\bullet = \<\Sigma, \Gamma^\bullet\>$. $\Gamma^\bullet$ is called the theory generated by $\Gamma$.
\end{definition}

\begin{definition}[Theory presentations \cite{meseguer:lc87}]
Let us consider the entailment system $\< \mathsf{Sign}, \mathbf{Sen}, \{\vdash_{\Sigma}\}_{\Sigma \in |\mathsf{Sign}|} \>$. Then, its category of theories is a structure $\< \mathcal{O}, \mathcal{A} \>$ (generally denoted by $\mathsf{Th}$) such that:
\begin{itemize}[$-$] 
\item $\mathcal{O} = \left\{ \< \Sigma, \Gamma \>\ |\ \Sigma \in |\mathsf{Sign}|\ \mbox{and}\ \Gamma \subseteq \mathbf{Sen}(\Sigma) \right\}$, and
\item $\mathcal{A} = \left\{ \sigma: \<\Sigma, \Gamma\> \to \<\Sigma', \Gamma'\> \ {\Big |}\ \< \Sigma, \Gamma \>, \< \Sigma', \Gamma' \> \in \mathcal{O}, \right.$

\hfill $\left. \sigma: \Sigma \to \Sigma' \in ||\mathsf{Sign}||, \mbox{for all}\ \gamma \in \Gamma, \Gamma' \vdash_{\Sigma'} \mathbf{Sen}(\sigma)(\gamma) \right\}$.
\end{itemize}
\end{definition}


An \emph{institution} is defined in a similar way, by identifying a class of models and a family of \emph{semactic} consequence relations, instead of a family of syntactic consequence relations.

\begin{definition}[Institutions \cite{goguen:cmwlp84}]
\label{institution}
A structure of shape $\< \mathsf{Sign}, \mathbf{Sen}, \mathbf{Mod}, \{\models_{\Sigma}\}_{\Sigma \in |\mathsf{Sign}|} \>$ is an \emph{institution} if and only if it satisfies the following conditions:
\begin{itemize}[$-$]
\item $\mathsf{Sign}$ is a category of signatures,
\item $\mathbf{Sen}: \mathsf{Sign} \to \mathsf{Set}$ is a functor.
\item $\mathbf{Mod}: \mathsf{Sign}^\op \to \mathsf{Cat}$ is a functor.
\item $\{\models_{\Sigma}\}_{\Sigma \in |\mathsf{Sign}|}$, where $\models_{\Sigma} \subseteq |\mathbf{Mod}(\Sigma)| \times \mathbf{Sen}(\Sigma)$, is a family of binary relations,
\end{itemize}
\noindent and for any $\sigma: \Sigma \to \Sigma' \in ||\mathsf{Sign}||$, $\Sigma$-sentence $\phi \in \mathbf{Sen}(\Sigma)$ and $\Sigma'$-model $\mathcal{M}' \in |\mathbf{Mod}(\Sigma)|$, the following
$\models$-invariance condition holds:
$$\mathcal{M}' \models_{\Sigma'} \mathbf{Sen}(\sigma)(\phi)\quad \mbox{
  iff }\quad \mathbf{Mod}(\sigma)(\mathcal{M}')
  \models_{\Sigma} \phi\ .$$
\end{definition}

Roughly speaking, the last condition above says that \emph{the notion of truth is invariant with respect to the change of notation (non-logical symbols)}. Given $\Sigma \in |\mathsf{Sign}|$ and $\Gamma \subseteq \mathbf{Sen} (\Sigma)$, $\mathbf{Mod} (\Sigma, \Gamma)$ denotes the full subcategory of $\mathbf{Mod} (\Sigma)$ determined by those models $\mathcal{M} \in |\mathbf{Mod} (\Sigma)|$ such that $\mathcal{M} \models_\Sigma \gamma$, for all $\gamma \in \Gamma$. The relation $\models^\Sigma$ between sets of formulae and formulae is defined in the following way: given $\Sigma \in |\mathsf{Sign}|$, $\Gamma \subseteq \mathbf{Sen} (\Sigma)$ and $\alpha \in \mathbf{Sen} (\Sigma)$, $\Gamma \models_\Sigma \alpha\quad$ if and only if $\quad \mathcal{M} \models_\Sigma \alpha$, for all $\mathcal{M} \in |\mathbf{Mod} (\Sigma, \Gamma)|$.

Now, if we consider the definition of $\mathbf{Mod}$ extended to signatures and sets of sentences, we get a functor $\mathbf{Mod}: \mathsf{Th}^\op \to \mathsf{Cat}$ defined as follows: let $T = \<\Sigma, \Gamma\> \in |\mathsf{Th}|$, then $\mathbf{Mod} (T) = \mathbf{Mod} (\Sigma, \Gamma)$.

Now, from Defs.~\ref{institution}~and~\ref{entailment-system}, it is possible to give a definition of \emph{logic} by relating both its model-theoretic and proof-theoretic characterisations; a coherence between the semantic and syntactic relations is required, reflecting the soundness and completeness of standard deductive relations of logical systems.

\begin{definition}[Logic \cite{meseguer:lc87}]
\label{logic}
 A structure of shape $\< \mathsf{Sign}, \mathbf{Sen}, \mathbf{Mod}, \{\vdash_{\Sigma}\}_{\Sigma \in |\mathsf{Sign}|}, \{\models_{\Sigma}\}_{\Sigma \in |\mathsf{Sign}|} \>$ is a  \emph{logic} if and only if it satisfies the following conditions:
\begin{itemize}[$-$]
\item $\< \mathsf{Sign}, \mathbf{Sen}, \{\vdash_{\Sigma}\}_{\Sigma \in |\mathsf{Sign}|} \>$ is an entailment system, 
\item $\< \mathsf{Sign}, \mathbf{Sen}, \mathbf{Mod}, \{\models_{\Sigma}\}_{\Sigma \in |\mathsf{Sign}|} \>$ is an institution, and 
\item the following \emph{soundness} condition is satisfied: for any $\Sigma \in |\mathsf{Sign}|$, $\phi \in \mathbf{Sen}(\Sigma)$, $\Gamma \subseteq \mathbf{Sen}(\Sigma)$: $\Gamma \vdash_\Sigma \phi\quad \mbox{implies}\quad \Gamma \models_\Sigma \phi\ .$
\end{itemize}
A logic is \emph{complete} if, in addition, the following condition is also satisfied: for any $\Sigma \in |\mathsf{Sign}|$, $\phi \in \mathbf{Sen}(\Sigma)$, $\Gamma \subseteq \mathbf{Sen}(\Sigma)$: $\Gamma \models_\Sigma \phi\quad\mbox{implies}\quad\Gamma \vdash_\Sigma \phi \ .$
\end{definition}

Given an entailment system $\<\mathsf{Sign}^{\sf L}, \mathbf{Sen}^{\sf L}, \{\vdash^{\sf L}_\Sigma\}_{\Sigma \in |\mathsf{Sign}^{\sf L}|}\>$, or an institution $\<\mathsf{Sign}^{\sf L}, \mathbf{Sen}^{\sf L}, \mathbf{Mod}^{\sf L}, \{\models^{\sf L}_\Sigma\}_{\Sigma \in |\mathsf{Sign}^{\sf L}|}\>$, for the logic ${\sf L}$, the structure containing it's first two components (i.e. $\<\mathsf{Sign}^{\sf L}, \mathbf{Sen}^{\sf L}\>$) will be referred to as the language of ${\sf L}$.

\section{Logical languages underlying the framework}
\label{languages}
In this section we review Equational Logic by formulating it within the theory of institutions, and present an extension of the Elementary Theory of Relation with reflexive and transitive closure of relational terms. These two logical languages constitute the formal background underlying the framework used in the forthcoming sections for formalising classes of relational models.

\subsection* {Equational logic}
\emph{Equational logic} \cite{henkin:amm-84_8,enderton72} has been studied for a long time as it is the best suited tool for characterising, and studying, the behaviour of a set of functions. This is because equational logic has the minimum logical structure\footnote{The term \emph{logical structure} is used to denote the invariant aspects of the interpretation of symbols in terms of mathematical structures. In the case of equational logic the only logical symbol is the equality ``$=$''; in first-order logic, the logical symbols are the boolean connectives ``$\neg$'', ``$\lor$'', ``$\exists$'' and the remaining operators that can be defined in terms of them.} a logical system can have in order to be used as the formal specification of the behaviour of functions. We will extend the traditional definition of equational logic by adding, as part of the signature, a set of extralogical predicate symbols.

\begin{definition}[The language of Equational Logic \cite{henkin:amm-84_8}]
\label{eq}
The language of \emph{Equational Logic} is a structure $\langle {\sf Sign^{Eq}}, {\bf Sen}^{\sf Eq} \rangle$ (denoted ${\sf Eq}$ for short) such that: 
\begin{itemize}[$-$]
\item ${\sf Sign^{Eq}} = \langle \mathcal{O}, \mathcal{A} \rangle$ where:
\begin{itemize}[$-$]
\item $\mathcal{O}$ is the class of structures $\langle \{C_j\}_{j \in J}, \{f_i\}_{i \in I}, \{P_k\}_{k \in K} \rangle$ where $I$, $J$, and $K$ are contable sets, and
\item $\mathcal{A} = \left\{ \langle\sigma_C, \sigma_f, \sigma_P\rangle: \langle C, F, P \rangle \to \langle C', F', P' \rangle \ {\Big |}\ \right.$

\hfill $\left. \sigma_C: C \to C', \sigma_F: F \to F', \sigma_P: P \to P' \mbox{ are total functions} \right\}$
\end{itemize}
\item $\mathbf{Sen}^{\sf Eq}: {\sf Sign^{Eq}} \to \mathbf{Set}$ is defined as follows:
\begin{itemize}[$-$]
\item let $\Sigma = \langle C, F, P \rangle \in |{\sf Sign^{Eq}}|$, then $\mathit{Term}(\Sigma)$ is the smallest set satisfying:
\begin{itemize}[$-$]
\item $C \subseteq \mathit{Term}(\Sigma)$, and
\item for all $f \in F$, $t_1, \ldots t_{\mathit{ar}(f)} \in \mathit{Term}(\Sigma)$, then $f (t_1, \ldots t_{\mathit{ar}(f)}) \in \mathit{Term}(\Sigma)$, and
\end{itemize}
\noindent and ${\bf Sen}^{\sf Eq}(\Sigma) = \{ t = t' \ | \ t, t' \in \mathit{Term}(\Sigma)\} \cup \{ P(t_1, \ldots, t_{\mathit{ar}(P)}) \ | \ t_1, \ldots, t_{\mathit{ar}(f)} \in \mathit{Term}(\Sigma)\}$.
\item let $\Sigma = \<C, F, P\>$ and $\Sigma' = \<C', F', P'\>$ and $\sigma = \<\sigma_C, \sigma_F, \sigma_P\>: \Sigma \to \Sigma' \in ||{\sf Sign^{Eq}}||$, then we define $\sigma^*: \mathit{Term}(\Sigma) \to \mathit{Term}(\Sigma')$ as follows:
\begin{itemize}[$-$]
\item for all $c \in C$, $\sigma^* (c) = {\sigma_C} (c)$, and
\item for all $f (t_1, \ldots, t_{\mathit{ar}(f)}) \in Term(\Sigma)$, \\
$\sigma^* (f (t_1, \ldots, t_{\mathit{ar}(f)})) = {\sigma_F} (f) (\sigma^* (t_1), \ldots, \sigma^* (t_{\mathit{ar}(f)}))$.
\end{itemize}
Then, we define ${\bf Sen}^{\sf Eq}(\sigma)$ as follows:
\begin{itemize}[$-$]
\item ${\bf Sen}^{\sf Eq}(\sigma)(t = t') = \sigma^* (t) = \sigma^* (t')$, and 
\item ${\bf Sen}^{\sf Eq}(\sigma)(P (t_1, \ldots, t_{\mathit{ar}(P)})) = {{\sigma_P}_k} (P) (\sigma^* (t_1), \ldots, \sigma^* (t_{\mathit{ar}(P)}))$.
\end{itemize}
\end{itemize}
\end{itemize}
Given $\Sigma = \langle C, F, P \rangle \in |{\sf Sign}^{\sf Eq}|$, we use $f \in \Sigma$ as a short for $f \in C$, $f \in F$, or $f \in P$. Also, we assume the existence of a function $\mathit{ar}: \Sigma \to \NAT$ such that if $\Sigma \in |{\sf Sign}|$, then for every symbol $f \in \Sigma$, $\mathit{ar}(f)$ is the number of arguments to which $f$ is supposed to apply.
\end{definition}
%
%
\begin{fact}
\label{fact:eq-entail}
Let ${\sf Eq} = \<\mathsf{Sign}^{\sf Eq}, \mathbf{Sen}^{\sf Eq}\>$ be the language of Def.~\ref{eq}, then the structure $\<\mathcal{L}^{\sf Eq}, \{\vdash^{\sf Eq}_\Sigma\}_{\Sigma \in |\mathsf{Sign}^{\sf Eq}|}\>$, where $\vdash^{\sf Eq} \subseteq 2^{\mathbf{Sen}^{\sf Eq}} \times \mathbf{Sen}^{\sf Eq}$ is the standard deduction relation for equational logic, is an entailment system.
\end{fact}
%

\subsection*{Elementary theory of relations with closure}
\label{etrc}
An outstanding effort to create an algebra in which logical reasoning can be carried out is due to Charles Sanders Peirce \cite{peirce:maas-9}. Peirce's work was deeply influenced by De Morgan's ``fourth memoir'' \cite{demorgan+:tcps-10}, where he sketched the theory of dyadic relations under the name ``the logic of relations''. This effort gave birth to the algebra of binary relations, originally as an attempt to obtain an algebraization of first-order predicate logic. It was in \cite{peirce:peirce1883} where Peirce gave the \emph{algebras of binary relations} its final shape, at that time, under the name ``the logic of relatives''. After that, Peirce's system for the algebras of binary relations was extensively developed by Schr\"{o}der in \cite{schroder1895}.

In \cite{tarski:jsl-6_3}, Tarski calls our attention to the fact that there was almost no research being carried out in the field until Whitehead and Russell \cite{whitehead27} included the algebras of binary relations in the whole of logic. It was he who committed to the development of a calculus for relations and, along the way, introduced the \emph{elementary theory of (binary) relations} as a logical formalisation of the algebra of binary relations. It was also Tarski who, in \cite{tarski:cm58}, introduced a predicate logic with infinitely long expressions, later developed by many authors like Barwise \cite{barwise:jsl-34_2} and Karp \cite{karp64}, who thoroughly studied many variants of this infinitary logical system. As Goldblatt points out in \cite{goldblatt:lncs-130}, it was Engeler who, in \cite{engeler:mst-1_2}, first used infinitely long formulae to provide formal meaning to the construct of iteration in programming languages. We will use it in an analogous way as the means for restricting the possible interpretations of the Kleene closure operator to those models in which it is interpreted as the reflexive and transitive closure.

We name \emph{elementary theory of (binary) relations with closure} the logical system obtained by enriching Tarski's \emph{elementary theory of (binary) relations} \cite{tarski:jsl-6_3} with the relational operator of reflexive and transitive closure (``${\ }^*$''), denumerable infinite disjunction (``$\bigvee$'') and denumerable infinite conjunction (``$\bigwedge$''). From now on we will omit the reference to the term binary as it is the only kind of relation formalised in Tarski's language.

\begin{definition}[Elementary theory of relations with closure]
\label{etr*}
The language of \emph{Elementary theory of relations with closure} is defined as a structure $\langle {\sf Sign^{ETR*}}, {\bf Sen}^{\sf ETR*} \rangle$ (denoted ${\sf ETR*}$ for short) such that: 
\begin{itemize}[$-$]
\item ${\sf Sign^{ETR*}} = \langle \mathcal{O}, \mathcal{A} \rangle$ where:
\begin{itemize}[$-$]
\item $\mathcal{O} = \{\{R_i\}_{i \in \mathcal{I}}\ |\ |\mathcal{I}|\ = \omega\}$, and
\item $\mathcal{A} = \left\{ \sigma: \{R_i\}_{i \in \mathcal{I}} \to \{R'_i\}_{i \in \mathcal{I}'} \ {\Big |}\ \sigma: \mathcal{I} \to \mathcal{I}' \mbox{ is a total function.} \right\}$;
\end{itemize}
\item $\mathbf{Form}^{\sf ETR*}: {\sf Sign^{ETR*}} \to \mathbf{Set}$ is defined as follows:
\begin{itemize}[$-$]
\item let $\Sigma = \{R_i\}_{i \in \mathcal{I}} \in |\mathsf{Sign}^{\sf ETR*}|$ then, $\mathit{Term}(\Sigma)$ is the smallest set satisfying:
\begin{itemize}[$-$]
\item $\Sigma \cup \set{\azero, \aunit, \aid} \subseteq \mathit{Term}(\Sigma)$, and
\item if $R, S \in \mathit{Term}(\Sigma)$, then $\set{R \ajoin S, R \ameet S, \acompl{R}, R \acompo S, \aconv{R}, \aclosure{R}} \subseteq \mathit{Term}(\Sigma)$.
\end{itemize} 
If $\Sigma = \{R_i\}_{i \in \mathcal{I}}$ is a signature $|{\sf Sign^{ETR*}}|$ and $\mathcal{X}$ is a countable set of flexible symbols for individuals (i.e., variable symbols) then $\mathit{AtForms}(\Sigma, \mathcal{X})$ is the smallest set satisfying:
\begin{itemize}[$-$]
\item if $x, y \in \mathcal{X}$ and $R \in \mathit{Term}(\Sigma)$, then $x\; R\; y \in \mathit{AtForms}(\Sigma, \mathcal{X})$, and
\item if $R, S \in \mathit{Term}(\Sigma)$, then $R = S \in \mathit{AtForms}(\Sigma, \mathcal{X})$.
\end{itemize}
Then, ${\bf Form}^{\sf ETR*}(\Sigma)$ then defined as the smallest set satisfying:
\begin{itemize}[$-$]
\item $\mathit{AtForm}(\Sigma, \mathcal{X}) \subseteq {\bf Form}^{\sf ETR*}(\Sigma)$, 
\item if $\{\alpha\} \cup S \subseteq {\bf Form}^{\sf ETR*}(\Sigma)$ and $x \in \mathcal{X}$, then $\set{\neg f, \bigvee S, (\exists x)(\alpha)} \subseteq {\bf Form}^{\sf ETR*}(\Sigma)$.\footnote{The rest of the propositional operators, such as denumerable infinite conjunction ($\bigwedge$) and implication ($\Longrightarrow$), are defined in terms of the negation ($\neg$) and denumerable infinite disjunction ($\bigvee$) operators as usual. The universal quantifier ($\forall$) is defined as the dual of the existential quantifier.}
\end{itemize}
\item if $\sigma: \{R_i\}_{i \in \mathcal{I}} \to \{R'_i\}_{i \in \mathcal{I}'} \in ||\mathsf{Sign}^{\sf ETR*}||$ then, $\sigma^\star: \mathit{Term}(\{R_i\}_{i \in \mathcal{I}}) \to \mathit{Term}(\{R'_i\}_{i \in \mathcal{I}'})$ is defined as follows:
\[
\begin{array}{rcl}
\sigma^\star(R_j ) & = & R'_{\sigma(j)} \mbox{, for all $R_j \in \{R_i\}_{i \in \mathcal{I}}$},\\
\sigma^\star(\azero) & = & \azero,\\
\sigma^\star(S \ajoin T) & = & \sigma^\star(S) \ajoin \sigma^\star(T),\\
\sigma^\star(\aunit) & = & \aunit,\\
\sigma^\star(S \ameet T) & = & \sigma^\star(S) \ameet \sigma^\star(T),\\
\sigma^\star(\acompl{S}) & = & \acompl{\sigma^\star(S)},\\
\sigma^\star(S \acompo T) & = & \sigma^\star(S) \acompo \sigma^\star(T),\\
\sigma^\star(\aconv{S}) & = & \aconv{\sigma^\star(S)}, \\
\sigma^\star(\aclosure{S}) & = & \aclosure{\sigma^\star(S)}.
\end{array}
\]
\noindent and if $\alpha \in \mathbf{Form}^{\sf ETR*}(\Sigma)$, we define the function ${\bf Form}^{\sf ETR*}(\sigma): {\bf Form}^{\sf ETR*}(\Sigma) \to {\bf Form}^{\sf ETR*}(\Sigma')$ as follows:
\[
\begin{array}{l}
{\bf Form}^{\sf ETR*}(\sigma)(x\; R\; y) = x\; \sigma^\star(R)\; y,\\
{\bf Form}^{\sf ETR*}(\sigma)(R = S) = \sigma^\star(R) = \sigma^\star(S),\\
{\bf Form}^{\sf ETR*}(\sigma)(\neg \varphi) = \neg {\bf Form}^{\sf ETR*}(\sigma)(\varphi),\\
{\bf Form}^{\sf ETR*}(\sigma)(\bigvee S) = \bigvee \{ {\bf Form}^{\sf ETR*}(\sigma)(\varphi)\ |\ \varphi \in S\},\\
{\bf Form}^{\sf ETR*}(\sigma)((\exists x)\varphi) = (\exists x){\bf Form}^{\sf ETR*}(\sigma)(\varphi).
\end{array}
\]
\end{itemize}
\end{itemize}
As usual, $\mathbf{Sen}^{\sf ETR*}: {\sf Sign^{ETR*}} \to \mathbf{Set}$ is defined as the restriction of $\mathbf{Form}^{\sf ETR*}$ to formulae with no free variables, referred to as a \emph{sentences}.
\end{definition}

The reader should note the closeness between the sentential fragment of elementary theory of relations with closure, presented before, and $L_{{\omega_1}, \omega}$ \cite[Sec.~11.4]{karp64} (i.e., first-order predicate logic with equality with denumerable infinitely long formulae). As we mentioned before, it was Pierce who introduced the class of algebras of binary relations as the target semantics of the language for describing relations and the more modern name of \emph{proper relation algebras} was introduced by Tarski. The next definition extends this class of algebras with reflexive and transitive closure.

\begin{definition}[Proper closure relation algebras]
\label{def:pcra}
A \emph{proper closure relation algebra} on a set $U$ (usually referred to as the base set of the algebra) is a structure $\< U, A, \pjoin, \pmeet, \pcomple, \pzero, \punit{E}, \pcompo, \pconve, \pid, \pclosure{} \>$ in which $A$ is a set of binary relations on $U$, $\pjoin$, $\pmeet$ and $\pcompo$ are binary operations, $\pcomple$ and $\pconve$ are unary operations and $\pzero$, $\punit{E}$ and $\pid$ are distinguished elements of $A$ satisfying:
\begin{itemize}[$-$]
\item $\pzero \in A$ is the empty relation on the set $U$,
\item $A$ is closed under $\pjoin$ (i.e. set union),
\item $\punit{E} \in A$ and $\bigcup_{r \in A} r \subseteq \punit{E}$,
\item $A$ is closed under $\pmeet$ (i.e. set intersection),
\item $A$ is closed under $\pcomple$ (i.e. set complement with respect to $\punit{E}$),
\item $\pid \in A$ is the identity relation on the set $U$.
\item $A$ is closed under $\pcompo$ (i.e. relation composition), defined as 

$x \pcompo y  = \setof{\pair{a}{b} \in U \times U}{ (\exists c)(\pair{a}{c} \in x \land \pair{c}{b} \in y)}$,
\item $A$ is closed under $\pconve$ (i.e. relation transposition), defined as: 

$\pconv{x} = \setof{\pair{a}{b} \in U \times U}{\pair{b}{a} \in x}$,
\item A is closed under $\pclosure{}$ (i.e. relation reflexive and transitive closure), defined as: 

$\pclosure{r} = \bigcup_{0 \leq i} r^{i}$, where $r^{0} = \pid$, and $r^{n+1} = r \pcompo r^{n}$.
\end{itemize}
A proper closure relation algebra is called \emph{full} if $A = U \times U$.
\end{definition}

\begin{proposition}
\label{prop:uniquefull}
Let $U$ a set, then there exists exactly one full proper closure relation algebra of cardinality $2^{U \times U}$ up-to isomorphisms.
\end{proposition}
\begin{proof}
It is easy to see that given $U$ and $U'$ sets, such that $|U| = |U'|$, any bijection $h: U \to U'$ can be extended to an isomorphism between the proper closure relation algebras over $U$ and $U'$.
\end{proof}

\begin{definition}[Models over proper closure relation algebras]
\label{def:etr*-models}
Let $U$ be a set, $\Sigma = \{R_i\}_{i \in \mathcal{I}} \in |\mathsf{Sign}^{\sf ETR*}|$ and $\mathcal{A}$ a proper closure relation algebra, then a \emph{model} for $\Sigma$ is a structure $\< \mathcal{A}, \{R^\mathcal{A}_i\}_{i \in \mathcal{I}} \>$.

Let $\{R_i\}_{i \in \mathcal{I}}, \{R'_i\}_{i \in \mathcal{I}'} \in |\mathsf{Sign}^{\sf ETR*}|$ be relational signatures and $\sigma: \{R_i\}_{i \in \mathcal{I}} \to \{R'_i\}_{i \in \mathcal{I}'} \in ||\mathsf{Sign}^{\sf ETR*}||$, $\mathbf{Mod}^{\sf ETR*}\(\sigma^\op\) \(\<\mathcal{A}, \{{R'_i}^\mathcal{A}\}_{i \in \mathcal{I}'}\>\) = \<\mathcal{A}, \left\{{\sigma \(R_i\)}^\mathcal{A}\right\}_{i \in \mathcal{I}}\>$.
\end{definition}

\begin{definition}[Satisfaction relation for the elementary theory of relations with closure]
\label{def:etr*-satisfaction}
Let $U$ be a set, $\mathcal{X}$ is a countable set of flexible symbols for individuals, $\Sigma = \{R_i\}_{i \in \mathcal{I}} \in |\mathsf{Sign}^{\sf ETR*}|$ and $\mathcal{M} = \<\mathcal{A}, \{R^\mathcal{A}_i\}_{i \in \mathcal{I}}\> \in |\mathbf{Mod}^{\mathsf{ETR*}} (\Sigma)|$, then we define the satisfaction relation $\models^{\sf ETR*}_\Sigma \subseteq |\mathbf{Mod}^{\sf ETR*}(\Sigma)| \times \mathbf{Sen}^{\sf ETR*}(\Sigma)$ as:
\begin{center}
$\mathcal{M} \models^{\sf ETR*}_\Sigma \alpha$ if and only if $\<\mathcal{M}, \emptyset\> \models^{\sf ETR*}_\Sigma \alpha$
\end{center}
\noindent with $\models^{\sf ETR*}_\Sigma \subseteq |\mathbf{Mod}^{\sf ETR*}(\Sigma)| \times \mathbf{Form}^{\sf ETR*}(\Sigma)$ defined as follows:
\[
\begin{array}{l}
\<\mathcal{M}, v\> \models^{\sf ETR*}_\Sigma x\ R\ y \text{ iff } \<v (x), v (y)\> \in m_\mathcal{M} (R)\\
\<\mathcal{M}, v\> \models^{\sf ETR*}_\Sigma R = S \text{ iff } \mbox{for all $x, y \in U$},  \<x, y\> \in m_\mathcal{M}(R) \text{ iff } \<x, y\> \in m_\mathcal{M}(S), \\
\<\mathcal{M}, v\> \models^{\sf ETR*}_\Sigma \neg \alpha \text{ iff } \<\mathcal{M}, v\> \models^{\sf ETR*}_\Sigma \alpha \mbox{ does not hold},\\
\<\mathcal{M}, v\> \models^{\sf ETR*}_\Sigma \bigvee S \text{ iff } \mbox{ there exists $\alpha \in S$, } \<\mathcal{M}, v\> \models^{\sf ETR*}_\Sigma \alpha, \\
\<\mathcal{M}, v\> \models^{\sf ETR*}_\Sigma (\exists x)\alpha \text{ iff } \mbox{ there exists $u \in U$,}  \<\mathcal{M}, v|^u_x\> \models^{\sf ETR*}_\Sigma \alpha.
\end{array}
\]
\noindent where $m_\mathcal{M}: \mathit{Term}(\Sigma) \to |\mathcal{A}|$ is defined as follows:
\[
\begin{array}{l}
m_\mathcal{M} (R) = R^\mathcal{A}, \mbox{ for all $R \in \Sigma$.}\\
m_\mathcal{M} (\azero) = \emptyset\\
m_\mathcal{M} (S \ajoin T) = m_\mathcal{M} (S) \cup m_\mathcal{M} (T)\\
m_\mathcal{M} (\aunit) = \punit{E}\\
m_\mathcal{M} (S \ameet T) = m_\mathcal{M} (S) \cap m_\mathcal{M} (T)\\
m_\mathcal{M} (\acompl{S}) = \overline{m_\mathcal{M} (S)}\\
m_\mathcal{M} (\aid) = \mathit{Id}\\
m_\mathcal{M} (S \acompo T) = m_\mathcal{M} (S) \circ m_\mathcal{M} (T)\\
m_\mathcal{M} (\aconv{S}) = {m_\mathcal{M} (S)}^t\\
m_\mathcal{M} (\aclosure{S}) = \pclosure{m_\mathcal{M} (S)}
\end{array}
\]
Given $\Gamma \subseteq \mathbf{Sen}^{\sf ETR*} (\Sigma)$ and $\mathcal{M} \in |\mathbf{Mod}^{\mathsf{ETR*}} (\Sigma)|$, $\mathcal{M} \models^{\sf ETR*}_\Sigma \Gamma$ if and only if for all $\alpha \in \Gamma$, $\mathcal{M} \models^{\sf ETR*}_\Sigma \alpha$. 

Let $\{\alpha\} \cup \Gamma \subseteq \mathbf{Sen}^{\sf ETR*} (\Sigma)$, $\alpha$ is a semantic consequence of $\Gamma$ (denoted as $\Gamma \models^{\sf ETR*}_\Sigma \alpha$), if for all $\mathcal{M} \in |\mathbf{Mod}^{\mathsf{ETR*}} (\Sigma)|$, $\mathcal{M} \models^{\sf ETR*}_\Sigma \Gamma$ implies $\mathcal{M} \models^{\sf ETR*}_\Sigma \alpha$.

Finally, if $\alpha \in \mathbf{Sen}^{\sf ETR*} (\Sigma)$, $\alpha$ is said to be \emph{valid} (denoted as $\models^{\sf ETR*}_\Sigma \alpha$), if for all $\mathcal{M} \in |\mathbf{Mod}^{\mathsf{ETR*}} (\Sigma)|$, $\mathcal{M} \models^{\sf ETR*}_\Sigma \alpha$. 
\end{definition}

\begin{lemma}
\label{lemma:full}
Let $U$ be a set, $\mathcal{X}$ is a countable set of flexible symbols for individuals, $\Sigma = \{R_i\}_{i \in \mathcal{I}} \in |\mathsf{Sign}^{\sf ETR*}|$, $\mathcal{M} = \<\mathcal{A}, \{R^\mathcal{A}_i\}_{i \in \mathcal{I}}\> \in |\mathbf{Mod}^{\mathsf{ETR*}} (\Sigma)|$, with $\mathcal{A}$ a proper closure relation algebra on a set $U$, $\mathcal{M}^\mathit{full} = \<\mathcal{A}^\mathit{full}, \{R^\mathcal{A}_i\}_{i \in \mathcal{I}}\> \in |\mathbf{Mod}^{\mathsf{ETR*}} (\Sigma)|$, with $\mathcal{A}^\mathit{full}$ the full proper closure relation algebra on a set $U$, and $\alpha \in \mathbf{Form}^{\sf ETR*} (\Sigma)$. Then, for all $v: \mathcal{X} \to U$, if $\<\mathcal{M}, v\> \models^{\sf ETR*}_\Sigma \alpha$ then $\<\mathcal{M}^\mathit{full}, v\> \models^{\mathsf{ETR*}}_\Sigma \alpha$.
\end{lemma}
\begin{proof}
The proof is trivial by induction on the structure of the formula $\alpha$, requiring also to prove that if $\<\mathcal{M}, v\> \models^{\sf ETR*}_\Sigma x\ R\ y$ then $\<\mathcal{M}^\mathit{full}, v\> \models^{\mathsf{ETR*}}_\Sigma x\ R\ y$, by induction on the structure of the relational terms.
\end{proof}

\begin{proposition}
Let $U$ be a set, $\mathcal{X}$ is a countable set of flexible symbols for individuals, $\Sigma = \{R_i\}_{i \in \mathcal{I}} \in |\mathsf{Sign}^{\sf ETR*}|$, $\mathcal{M} = \<\mathcal{A}, \{R^\mathcal{A}_i\}_{i \in \mathcal{I}}\> \in |\mathbf{Mod}^{\mathsf{ETR*}} (\Sigma)|$, with $\mathcal{A}$ a proper closure relation algebra on a set $U$, $\mathcal{M}^\mathit{full} = \<\mathcal{A}^\mathit{full}, \{R^\mathcal{A}_i\}_{i \in \mathcal{I}}\> \in |\mathbf{Mod}^{\mathsf{ETR*}} (\Sigma)|$, with $\mathcal{A}^\mathit{full}$ the full proper closure relation algebra on a set $U$, and $\alpha \in \mathbf{Sen}^{\sf ETR*} (\Sigma)$. Then, if $\mathcal{M} \models^{\sf ETR*}_\Sigma \alpha$ then $\mathcal{M}^\mathit{full} \models^{\mathsf{ETR*}}_\Sigma \alpha$.
\end{proposition}
\begin{proof}
The proof follows directly from Lemma~\ref{lemma:full}.
\end{proof}

\begin{corollary}
\label{coro:full}
Let $U$ be a set, $\mathcal{X}$ is a countable set of flexible symbols for individuals, $\Sigma = \{R_i\}_{i \in \mathcal{I}} \in |\mathsf{Sign}^{\sf ETR*}|$, $\mathcal{M} = \<\mathcal{A}, \{R^\mathcal{A}_i\}_{i \in \mathcal{I}}\> \in |\mathbf{Mod}^{\mathsf{ETR*}} (\Sigma)|$, with $\mathcal{A}$ a proper closure relation algebra on a set $U$, $\mathcal{M}^\mathit{full} = \<\mathcal{A}^\mathit{full}, \{R^\mathcal{A}_i\}_{i \in \mathcal{I}}\> \in |\mathbf{Mod}^{\mathsf{ETR*}} (\Sigma)|$, with $\mathcal{A}^\mathit{full}$ the full proper closure relation algebra on a set $U$, and $\Gamma \subseteq \mathbf{Sen}^{\sf ETR*} (\Sigma)$. Then, if $\mathcal{M} \models^{\sf ETR*}_\Sigma \Gamma$ then $\mathcal{M}^\mathit{full} \models^{\mathsf{ETR*}}_\Sigma \Gamma$.
\end{corollary}

\begin{definition}[Calculus for the elementary theory of relations with closure]
\label{def:etr*-entail}
\noindent Calculus for first order operators:  let $\mathcal{X}$ be a denumerable set of symbols for individuals and $x, y \in \mathcal{X}$
$$
\begin{array}{ccc}
\mbox{
\AXC{$\Gamma, \alpha \vdash \alpha$}
\DP
}
&
,
&
\mbox{
\AXC{$\Gamma \vdash \alpha$}
\AXC{$\Gamma', \alpha \vdash \beta$}
\LL{[Cut]}
\BinaryInfC{$\Gamma, \Gamma' \vdash \beta$}
\DP
}
\end{array}
$$
$$
\mbox{
\AXC{$\{\Gamma_i \vdash \alpha_i\}_{i \in \mathcal{I}}$}
\LL{[$\bigwedge$-intro]}
\UnaryInfC{$\bigcup_{i \in \mathcal{I}} \Gamma_i \vdash \bigwedge \{\alpha_i\ |\ i \in \mathcal{I}\}$}
\DP
}
$$
$$
\mbox{
\AXC{$\Gamma \vdash \bigwedge S$}
\LL{[$\bigwedge$-elim]}\RL{$[\alpha \in S]$}
\UnaryInfC{$\Gamma \vdash \alpha$}
\DP
}
$$
$$
\mbox{
\AXC{$\Gamma \vdash \alpha$}
\LL{[$\bigvee$-intro]}\RL{[$\alpha \in S$]}
\UnaryInfC{$\Gamma \vdash \bigvee S$}
\DP
}
$$
$$
\mbox{
\AXC{$\Gamma \vdash \bigvee \{\alpha_i\ |\ i \in \mathcal{I}\}$}
\AXC{$\{\Gamma_i, \alpha_i \vdash \varphi\}_{i \in \mathcal{I}}$}
\LL{[$\bigvee$-elim]}
\BinaryInfC{$\Gamma \cup \bigcup_{i \in \mathcal{I}} \Gamma_i \vdash \varphi$}
\DP
}
$$
$$
\mbox{
\AXC{$\Gamma, \varphi \vdash \alpha$}
\AXC{$\Gamma', \varphi \vdash \neg \alpha$}
\LL{[$\neg$-intro]}
\BinaryInfC{$\Gamma, \Gamma' \vdash \neg \varphi$}
\DP
}
$$
$$
\mbox{
\AXC{$\Gamma \vdash \neg\neg\alpha$}
\LL{[$\neg$-elim]}
\UnaryInfC{$\Gamma \vdash \alpha$}
\DP
}
$$
$$
\mbox{
\AXC{$\Gamma \vdash \alpha$}
\LL{[$\forall$-intro]}\RL{$x \not\in \mathit{FreeVar}(\Gamma)$}
\UnaryInfC{$\Gamma \vdash (\forall x)\alpha$}
\DP
}
$$
$$
\mbox{
\AXC{$\Gamma \vdash (\forall x)\alpha$}
\LL{[$\forall$-elim]}\RL{$x \not\in \mathit{FreeVar}(t)$}
\UnaryInfC{$\Gamma \vdash \alpha|^t_x$}
\DP
}
$$
$$
\mbox{
\AXC{$\Gamma \vdash \alpha|^t_x$}
\LL{[$\exists$-intro]}\RL{$\mathit{Var}(t) \subseteq \mathit{FreeVar}(\alpha)/\{x\}$}
\UnaryInfC{$\Gamma \vdash (\exists x)\alpha$}
\DP
}
$$
$$
\mbox{
\AXC{$\Gamma \vdash (\exists x)\alpha$}
\AXC{$\Gamma', \alpha \vdash \beta$}
\LL{[$\exists$-elim]}\RL{$x \not\in \mathit{FreeVar}(\beta)$}
\BinaryInfC{$\Gamma, \Gamma' \vdash \beta$}
\DP
}
$$
\noindent Axioms for the relational operators: let $\mathcal{R}$ be a denumerable set of relational symbols and $R, S \in \mathcal{R}$
$$
\begin{array}{lcl}
\mbox{Ax.~1} & \quad & (\forall x,y)(\neg (x\ \azero\ y))\\ 
\mbox{Ax.~2} &\quad & (\forall x,y)(x\ R \ajoin S\ y \iff x\ R\ y \lor x\ S\ y)\\ 
\mbox{Ax.~3} & \quad & (\forall x,y)(x\ \aunit\ y)\\ 
\mbox{Ax.~4} &\quad & (\forall x,y)(x\ R \ameet S\ y \iff x\ R\ y \land x\ S\ y)\\ 
\mbox{Ax.~5} &\quad & (\forall x,y)(x\ \acompl{R}\ y \iff \neg (x\ R\ y))\\ 
\mbox{Ax.~6} & \quad & (\forall x)(x\ \aid\ x)\\ 
\mbox{Ax.~6'} & \quad & (\forall x, y, z)(x\ R\ y \land y\ \aid\ z \implies x\ R\ z)\\ 
\mbox{Ax.~7} &\quad & (\forall x,y)(x\ R \acompo S\ y \iff (\exists z)(x\ R\ z \land z\ S\ y))\\ 
\mbox{Ax.~8} &\quad & (\forall x,y)(x\ \aconv{R}\ y \iff y\ R\ x)\\ 
\mbox{Ax.~9} &\quad & (\forall x,y)(x\ \aclosure{R}\ y \iff x\ \bigvee \{R^{;i}\ |\ i \in \NAT\}\ y)\\
                       &          & \qquad\mbox{, where $R^{;0} = \aid$, and $R^{;n+1} = R \acompo R^{;n}$}\\ 
\mbox{Ax.~10} &\quad & R = S \iff (\forall x,y)(x\ R\ y \iff x\ S\ y)
\end{array}
$$
\end{definition}
\begin{theorem}[Elementary theory of relations with closure]
\label{etr*-is-logic}
Let ${\sf ETR*}$ be the language $\<\mathsf{Sign}^{\sf ETR*}, \mathbf{Sen}^{\sf ETR*}\>$ from Def.~\ref{etr*}, $\mathbf{Mod}^{\sf ETR*}: {\mathsf{Sign}^{\sf ETR*}}^\op \to \mathsf{Cat}$ be the functor from Def.~\ref{def:etr*-models}, $\{\vDash^{\sf ETR*}_{\Sigma}\}_{\Sigma \in |\mathsf{Sign}^{\sf ETR*}|}$ the family of consequence relations from Def.~\ref{def:etr*-satisfaction}, and $\{\vdash^{\sf ETR*}_{\Sigma}\}_{\Sigma \in |\mathsf{Sign}^{\sf ETR*}|}$ the family of entailment relations from Def.~\ref{def:etr*-entail}, then:
\begin{enumerate}
\item \label{thm:institution2} $\< \mathsf{Sign}^{\sf ETR*}, \mathbf{Sen}^{\sf ETR*}, \mathbf{Mod}^{\sf ETR*}, \{\models^{\sf ETR*}_{\Sigma}\}_{\Sigma \in |\mathsf{Sign}^{\sf ETR*}|}\>$ is an institution,
\item \label{thm:entailment2} $\< \mathsf{Sign}^{\sf ETR*}, \mathbf{Sen}^{\sf ETR*}, \{\vdash^{\sf ETR*}_{\Sigma}\}_{\Sigma \in |\mathsf{Sign}^{\sf ETR*}|}\>$ is an entailment system, and
\item \label{thm:sound+complete2} $\< \mathsf{Sign}^{\sf ETR*}, \mathbf{Sen}^{\sf ETR*}, \mathbf{Mod}^{\sf ETR*}, \{\vdash^{\sf ETR*}_{\Sigma}\}_{\Sigma \in |\mathsf{Sign}^{\sf ETR*}|}, \right.$\\
$\left. \qquad\{\models^{\sf ETR*}_{\Sigma}\}_{\Sigma \in |\mathsf{Sign}^{\sf ETR*}|}\>$ is a sound and complete logic.
\end{enumerate}
\end{theorem}  
\begin{proof}
The proofs of Parts~\ref{thm:institution2}~and~\ref{thm:entailment2} follow directly from Defs.~\ref{etr*},~\ref{def:etr*-models},~\ref{def:etr*-satisfaction}~and~\ref{def:etr*-entail} and are analogous to the many examples of definitions in the literature of institutions and entailment systems. 

A straightforward proof for Part~\ref{thm:sound+complete2} follows by observing that $L_{{\omega_1},\omega}$ is strong complete (a direct consequence of \cite[\S 11.2, Deduction Theorem 11.2.4 and \S 11.4, Completness Theorem 11.4.1]{karp64}) and that every formula involving complex relational operators can be reduced using Axioms~$1$~to~$10$ to a formula in $L_{{\omega_1},\omega}$.
\end{proof}


The previous result proves the strong completeness of the entailment relations $\vdash^{\sf ETR*}_{\Sigma}$ for the class of models determined by the satisfaction relation $\models^{\sf ETR*}_{\Sigma}$, for all $\Sigma \in |\mathsf{Sign}^{\sf ETR*}|$. An important note on this result that will be of use in the forthcoming sections is that $L_{{\omega_1},\omega}$-theories are of cardinality $\omega_1$ as it not only involves formulae of the first-order predicate logic, but also disjunctions / conjunctions over denumerable infinite sets of formulae \cite[Sec.~1.1]{karp64}/\cite[\S1]{engeler:ma-151}. In \cite{karp64}, $L_{{\omega_1},\omega}$ is an example of the definition of $(\alpha, \beta, o, \pi)$-languages \cite[Sec.~8.1]{karp64}, referred to as ``Systems of formulas of infinite length''. In $(\alpha, \beta, o, \pi)$-languages, with $\alpha$ and $o$ regular infinite cardinals, and $\beta$ and $\pi$ infinite cardinals\footnote{A cardinal $\kappa$ is said to be \emph{regular} if and only if every unbounded subset $C \subseteq \kappa$ has cardinality $\kappa$ \cite[Chap.~9, pp.~254]{enderton77}; in other words, it is equal to its own cofinality \cite[Chap.~9, pp.~257]{enderton77}.}, $\alpha$ is the upper bound\footnote{Upper bound, in this context, means the least cardinal, strictly bigger.} for the cardinality of conjunctions / disjunctions, $\beta$ for the amount of quantified variables, $o$ for the length of the terms and $\pi$ for the arity of predicates. The regularity of a cardinal, in the case of $\alpha$ being $\omega_1$, plays a central role in guarantying that length of the conjunctions / disjunctions is less than $\omega_1$ (i.e., at most $\omega$).

\begin{theorem}[Skolem-L{\"{o}}wenheim Theorem for Sets of Infinitary Sentences, \cite{karp64}, Thm.~10.3.8]
Let $\Gamma$ be a set of $(\alpha, \beta, o, \pi)$-sentences with $\alpha \leq |\Gamma|$ and $|\Gamma|\ \mathit{exp}\ \epsilon = |\Gamma|$ for all $\epsilon < \beta$. Then $\Gamma$ is satisfiable if and only if $\Gamma$ is satisfiable in a model of power at most $|\Gamma|$.
\end{theorem}

In the previous theorem, given cardinals $\alpha$ and $\beta$, $\alpha\ \mathit{exp}\ \beta$ refers to cardinal exponentiation, equivalent to $|\alpha^\beta|$. Then, applying the theorem to $L_{{\omega_1}, \omega}$, we obtain the following corollary, guarantying the existence of models of cardinality $\omega_1$ for theory presentations in $\mathsf{Th}^{\mathsf{ETR*}}$.

\begin{corollary}
\label{coro-skolem-lowenheim}
Let $\Gamma$ be a set of $(\omega_1, \omega, \omega, \omega)$-sentences (i.e. $\Gamma \subseteq L_{{\omega_1},\omega}$). Then $\Gamma$ is satisfiable if and only if $\Gamma$ is satisfiable in a model of power at most $\omega_1$.
\end{corollary}
\section{Relational semantics from a proof theory standpoint}
\label{relmodels}
In this section we present the main contribution of this article by defining a framework in which it is possible to provide a formal proof theoretical characterisation of classes of relational models like those used to provide semantics to many modal \cite{kripke:apf-16}, hybrid \cite{areces:phdthesis} and deontic logics \cite{aqvist:hpl01}, among others. Such relational models are generally referred to as Kripke structures \cite{kripke:ttm65}. Besides their many differences, derived from the semantic needs of the syntactic features of each of these logical frameworks, their underlying structures can be thought of as \emph{Labelled transition systems} where locations (usually called states or worlds) are considered to be places where a formula can be assigned a truth value, and transitions respond to the need of interpreting modal operators as specific traverses of the relational structure connecting locations.

\begin{definition}[Labeled transition system]
\label{def:lts}
A \emph{labeled transition system} is a structure $\<S, \{R_i\}_{i \in \mathcal{I}}\>$ such that for all $i \in \mathcal{I}$, $R_i \subseteq S \times S$.
\end{definition}

The generality of the previous definition ensures that subclasses satisfying specific properties can be obtained by formalising them within a logical language of choice. Consider, as an example, the class of Kripke structures for giving semantics to linear temporal logics \cite{pnueli:tcs-13_1}. Formulae are defined over a set of propositional variables $\{p_i\}_{i \in \mathcal{I}}$ and the relational structure has the shape $\<S, S_0, T, l\>$ with $S$ being a set of states, $S_0 \subseteq S$ a set of initial states, $T \subseteq S \times S$ a binary relation called the accessibility relation, and $l: S \to 2^{\{p_i\}_{i \in \mathcal{I}}}$ a labelling function assigning a subset of the propositions to each state in $S$. Notice that Def.~\ref{def:lts} provides no support for tagging a state as initial, or satisfying a subset of the propositions so such additional properties will have to be built on top of the notion of relation surpassing its role solely as accessibility relation. In general, a relational structure is extended to a \emph{relational model} by adding an explicit reference to the particular state (or structure of states) over which the formulae are supposed to be evaluated. In the case of linear temporal logic, given a Kripke structure $\mathcal{K} = \<S, S_0, T, l\>$, the satisfaction of a formula is defined over infinite sequences $[s_1, s_2, \ldots]$ such that:
\begin{inparaenum}[a.]
\item for all $i \in \NAT$, $s_i \in S$,
\item for all $i \in \NAT$, $1 < i$ implies $\<s_i, s_{i+1}\> \in T$, and
\item $s_1 \in S_0$. 
\end{inparaenum}

The rationale behind our contribution is that the classical fragment of the logic will be assigned semantics as an interpretation of the rigid non-logical symbols, and an interpretation of the flexible non-logical symbols, both formalised as equational theory presentation, while the operators determining the relational behaviour will be characterised by the models of a theory presentation in the elementary theory of relations with closure.

Figure~\ref{figure:diagram} depicts a graphical view of the framework we propose, as a guide to be used by the reader throughout this section.
\begin{figure*}[!t]%
\centering
\includegraphics[width=0.90\textwidth]{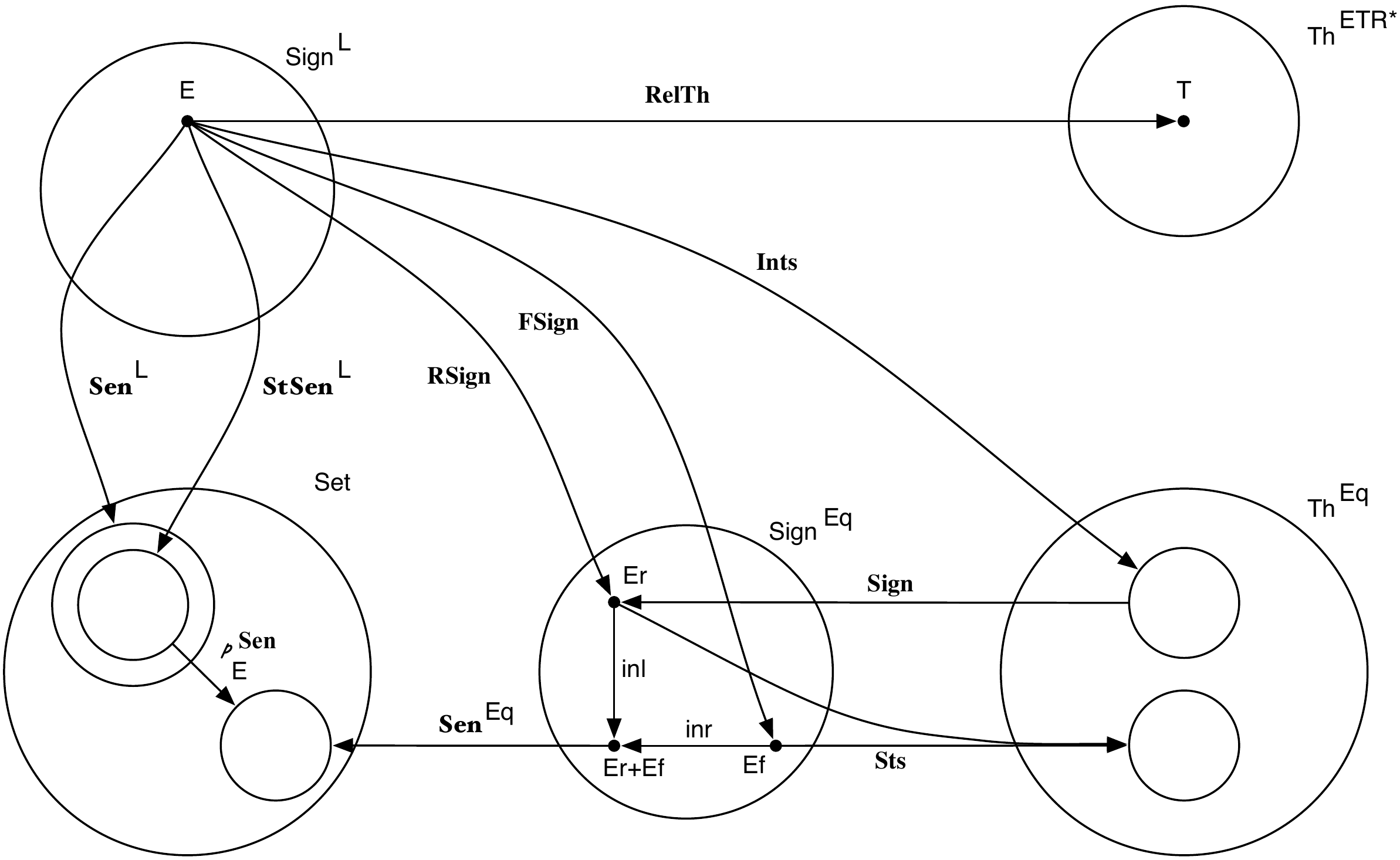}
\caption{Formal framework for axiomatising classes of relational models}\label{figure:diagram}
\end{figure*}

We have chosen first order dynamic logic (${\sf FODL}$ for short) \cite[Part~III, Chap.~11]{harel00} as a running example for the constructions presented in the paper, so let us start by summarising it's syntax and semantics as a quick reference for the features of the logical language. 

Signatures are structures of the form $\<C, F, P, A\>$ containing denumerable sets of symbols for constant, functions, predicates and atomic programs, respectively. Analogous to the presentation of any first order logical language, we start by defining the notion of term; let $\Sigma = \<C, F, P, A\>$ be a signature and $X$ a set of first order variable symbols, then $\mathit{TermFODL} (\Sigma)$ is the smallest set such that:
\begin{itemize}[$-$]
\item $C \cup X \subseteq \mathit{TermFODL} (\Sigma)$, and
\item if $f \in F$ and $\{t_1, \ldots, t_{\mathit{ar} (f)}\} \subseteq \mathit{TermFODL} (\Sigma)$, then $f (t_1, \ldots, t_2) \in \mathit{TermFODL} (\Sigma)$
\end{itemize}
Next, we mutually define formulae and programs, as the smallest sets $\mathit{FormFODL} (\Sigma)$ and $\mathit{PrgFODL} (\Sigma)$ such that:
\begin{itemize}[$-$]
\item if $t_1, t_2 \in \mathit{TermFODL} (\Sigma)$, then\\ $t_1 = t_2 \in \mathit{FormFODL} (\Sigma)$, 
\item if $p \in P$ and $\{t_1, \ldots, t_{\mathit{ar} (p)}\} \subseteq \mathit{TermFODL} (\Sigma)$, then $p (t_1, \ldots, t_2) \in \mathit{FormFODL} (\Sigma)$,
\item if $a \in \mathit{PrgFODL}(\Sigma)$, $x \in X$ and $\alpha, \beta \in \mathit{FormFODL} (\Sigma)$, then $\{\neg\alpha, \alpha\lor\beta, (\exists x)\alpha, \<a\>\alpha\} \subseteq \mathit{FormFODL} (\Sigma)$, 
\item $A \subseteq \mathit{PrgFODL} (\Sigma)$, and
\item if $\alpha \in \mathit{FormFODL}$ and $P, Q \mathit{PrgFODL} (\Sigma)$, then\\ $\{\alpha?, P+Q, P;Q, P^*\} \subseteq \mathit{PrgFODL} (\Sigma)$.
\end{itemize}

As it is well known, the language of regular programs over an alphabet of actions, like the one defined above, is expressive enough for defining the usual programming constructions:
\begin{itemize}[$-$]
\item $\mathit{if}\ \alpha\ \mathit{then}\ P\ \mathit{else}\ Q = \(\alpha?;P\)+\((\neg\alpha)?;Q\)$,
\item $\mathit{while}\ \alpha\ \mathit{do}\ P = \(\alpha?;P\)^*;(\neg\alpha)?$,
\end{itemize}

The semantics of ${\sf FODL}$ formulae over a signature $\<C, F, P, A\>$ and a set of first order variable symbols $X$ is given, as usual, in terms of an interpretation $\<S, \{\overline{c}\}_{c \in C}, \{\overline{f}\}_{f \in F}, \{\overline{p}\}_{p \in P}\>$ of the first order signature $\<C, F, P\>$, satisfying:
\begin{inparaenum}[1)]
\item $\overline{c} \in S$,
\item $\overline{f} \in [S^{\mathit{ar} (f)} \to S]$, for all $f \in F$, and
\item $\overline{p} \subseteq S^{\mathit{ar} (p)} \times S$, for all $p \in P$,
\end{inparaenum}
\noindent and states of a Kripke structure $\mathfrak{K} = \<W, W_0, \mathcal{L}\>$ such that:
\begin{inparaenum}[a)]
\item $W_0 \subseteq W$, and
\item $\mathcal{L}: W \to [X \to S]$.
\end{inparaenum}
Therefore, if $\alpha \in \mathit{FormFODL} (\<C, F, P, A\>)$, then $\alpha$ is satisfied by $w \in W$ in $\mathfrak{K}$ (written $\mathfrak{K}, w \models \alpha$) is inductively defined as follows:
\[
\begin{array}{l}
\mathfrak{K}, w \models t_1 = t_2 \text{ iff } m_{\<\mathfrak{K}, w\>} (t_1) = m_{\<\mathfrak{K}, w\>} (t_2)\\
\mathfrak{K}, w \models p (t_1, \ldots, t_{\mathit{ar} (p)}) \text{ iff } \<m_{\<\mathfrak{K}, w\>} (t_1), \ldots, m_{\<\mathfrak{K}, w\>} (t_{\mathit{ar} (p)})\> \in \overline{P}\\
\mathfrak{K}, w \models \neg\alpha \text{ iff } \mathfrak{K}, w \models \alpha \text{ does not hold}\\
\mathfrak{K}, w \models \alpha\lor\beta \text{ iff } \mathfrak{K}, w \models \alpha \text{ or } \mathfrak{K}, w \models \beta\\
\mathfrak{K}, w \models (\exists x) \alpha \text{ iff } \mbox{ there exists $w' \in W, s \in S$ such that:}\\
\hfill \mathcal{L} (w') = \mathcal{L} (w)[x \mapsto s] \mbox{ and } \mathfrak{K}, w' \models \alpha\\
\mathfrak{K}, w \models \<P\> \alpha \text{ iff } \mbox{ there exists $w' \in W$ such that:}\\
\hfill \<w, w'\> \in m_{\mathfrak{K}} (P) \text{ and } \mathfrak{K}, w' \models \alpha
\end{array}
\]
\noindent where $m_{\<\mathfrak{K}, w\>}$ and $m_{\mathfrak{K}}$ are defined as follows:\\
\[
\begin{array}{l}
m_{\<\mathfrak{K}, w\>} (c) = \overline{c}\\
m_{\<\mathfrak{K}, w\>} (x) = \mathcal{L} (w) (x)\\
m_{\<\mathfrak{K}, w\>} (f (t_1, \ldots, t_{\mathit{ar}(f)})) =  \overline{f} (m_{\<\mathfrak{K}, w\>} (t_1), \ldots, m_{\<\mathfrak{K}, w\>} (t_{\mathit{ar}(f)})).
\end{array}
\]
\[
\begin{array}{l}
m_{\mathfrak{K}} (a) = \overline{a}\\
m_{\mathfrak{K}} (\alpha?) = \{\<w, w\>\ |\ \mathfrak{K}, w \models \alpha\}\\
m_{\mathfrak{K}} (P;Q) = m_{\mathfrak{K}} (P) \circ m_{\mathfrak{K}} (Q)\\
m_{\mathfrak{K}} (P+Q) = m_{\mathfrak{K}} (P) \cup m_{\mathfrak{K}} (Q)\\
m_{\mathfrak{K}} (P^*) = \bigcup_{i \in \NAT} \(m_{\mathfrak{K}} (P)\)^{i}  \mbox{; where $r^{0} = \mathit{id}$, and $r^{n+1} = r \circ r^{n}$}
\end{array}
\]

From here on we assume that $\<\mathsf{Sign}^{\sf FODL}, \mathbf{Sen}^{\sf FODL}\>$ is the language, presented within the theory of general logics, for ${\sf FODL}$.

\subsection{On rigid and flexible symbols}
In most logical languages, the symbols defined in the signature are referred to as non-logical (or extralogical) and are known as rigid designators, or rigid symbols, as they designate the same object in every possible state. In contrast, a symbol is said to be a flexible designator (originally called flaccid designator, term coined by Saul Kripke in his 1970 lecture series at Princeton University, later published as the book \emph{Naming and Necessity} \cite[pp.~22]{kripke80}) when the object it designates depends on the specific state in which satisfaction is being evaluated. In first order languages the values of flexible symbols only range over individuals, while higher order formal languages might have flexible function and predicate symbols; higher order logic \cite{vanbenthem:hlfcs83}, \emph{abstract state machines} \cite{maibaum:ijcis-9_1} and relational databases \cite{codd:cacm-13_6}, are examples of systems making use of such higher order flexible symbols. Both flexible and rigid symbols are assigned values coming from the same domains; the only difference is whether they are interpreted over the frame or the state, hinting at a general and homogeneous view of what we will call \emph{interpretations} and \emph{states}.

In order to provide a proof-theoretical formalisation for relational semantics, we first pursue a formal characterisation of the values that symbols are to be assigned when they are interpreted. As we mentioned before, we are only interested in those values that can be designated by syntactic terms (i.e., values that can be named through algebraic terms). If we confine ourselves to first order logical languages, equational logic provides the means for completely axiomatising the behaviour of function symbols (and predicate symbols, if we consider the extension presented in Sec.~\ref{languages}), but higher order logical languages must rely on a more complex notion of interpretation and state, thus requiring the use of higher order equational logic \cite{meinke:tcs-100_2}, presented in full detail in Sec.~\ref{sec:hoeq}. Note that, by the way in which higher order equational logic is presented, the definitions and results presented in this section can be easily generalised to logical languages of higher order.

\begin{definition}
Let $L = \<\mathsf{Sign}^{\sf L}, \mathbf{Sen}^{\sf L}\>$ be a language then, the \emph{state sublanguage} of ${\sf L}$ is a structure $\<\mathbf{StSen}, \mathbf{RSign}, \mathbf{FSign}, \rho^{Sen}\>$ such that:
\begin{itemize}[$-$]
\item $\mathbf{StSen}^{\sf L}: \mathsf{Sign}^{\sf L} \to \mathsf{Set}$ is a subfunctor of $\mathbf{Sen}^{\sf L}$, 
\item $\mathbf{RSign}: \mathsf{Sign}^{\sf L} \to \mathsf{Sign}^{\sf Eq}$ and $\mathbf{FSign}: \mathsf{Sign}^{\sf L} \to \mathsf{Sign}^{\sf Eq}$ are functors, and
\item $\rho^{Sen}: \mathbf{StSen}^{\sf L} \nat \mathbf{Sen}^{\sf Eq} \circ (\mathbf{RSign}+\mathbf{FSign})$ is a natural transformation.
\end{itemize}
\end{definition}

The above definition aims to characterise the sublanguage of a logical language, containing those formulae whose evaluation only depends on the interpretation of the rigid symbols and the current assignment of values to the flexible symbols (i.e., that do not require the examination of any other state of the relational structure). As such, the functors $\mathbf{RSign}$ and $\mathbf{FSign}$ can be understand as classifiers of the symbols appearing in the $L$-signature, as rigid or flexible symbols, respectively.

\begin{example}[State formulae of first order dynamic logic]
\label{ex:state-sublanguage}
Let $X$ be a countable set of variables, $\<\mathsf{Sign}^{\sf FODL}, \mathbf{Sen}^{\sf FODL}\>$ be the language of first order dynamic logic as it was defined above, and $\Sigma = \<C, F, P, A\> \in |\mathsf{Sign}^{\sf FODL}|$, then we define $\<\mathbf{StSen}, \mathbf{RSign}, \mathbf{FSign}, \rho^{Sen}\>$ the state sublanguage of ${\sf FODL}$: 
\begin{itemize}[$-$]
\item $\mathbf{StSen}^{\sf FODL} (\Sigma)$ is all the atomic formulae of the form:
\begin{inparaenum}[1)]
\item $t_1 = t_2$, for all $t_1, t_2 \in \mathit{TermFODL} (\Sigma)$, and 
\item $p (t_1, \ldots, t_{\mathit{ar}(p)})$, for all $p \in P$ and $\{t_1, \ldots, t_{\mathit{ar}(p)}\} \in \mathit{TermFODL} (\Sigma)$.
\end{inparaenum} 
\item $\mathbf{RSign} (\Sigma) = \Sigma$. As usual the rigid symbols are those appearing in the $\mathsf{FODL}$ signature and interpreted in the same way in all the worlds in the model.
\item $\mathbf{FSign} (\Sigma) = \<X, \emptyset, \emptyset\>$. Flexible symbols are an equational signature containing only symbols that must be interpreted as individuals.
\item $\rho^{Sen}_\Sigma: \mathbf{StSen}^{\sf FODL} (\Sigma) \to \mathbf{Sen}^{\sf Eq} \circ (\mathbf{RSign}+\mathbf{FSign}) (\Sigma)$, is defined as follows:
\[
\begin{array}{l}
\rho^{\mathit{Sen}}_\Sigma (t = t') = \rho^{\mathit{Term}}_\Sigma (t) = \rho^{\mathit{Term}}_\Sigma (t')\\
\rho^{\mathit{Sen}}_\Sigma (p(t_1, \ldots, t_{\mathit{ar}(p)})) =  \mathit{in}_l (p) (\rho^{\mathit{Term}}_\Sigma (t_1), \ldots, \rho^{\mathit{Term}}_\Sigma (t_{\mathit{ar}(p)})) \mbox{, for all $p \in \Sigma$}
\end{array}
\]
\[
\begin{array}{l}
\rho^{\mathit{Term}}_\Sigma (c) = \mathit{in}_l (c) \mbox{, for all $c \in \Sigma$}\\
\rho^{\mathit{Term}}_\Sigma (x) = \mathit{in}_r (x) \mbox{, for all $x \in X$}\\
\rho^{\mathit{Term}}_\Sigma (f(t_1, \ldots, t_{\mathit{ar}(f)})) =  \mathit{in}_l (f)(\rho^{\mathit{Term}}_\Sigma (t_1), \ldots, \rho^{\mathit{Term}}_\Sigma (t_{\mathit{ar}(f)})) \mbox{, for all $f \in \Sigma$}
\end{array}
\]
\end{itemize}
\end{example}

\subsection{On interpretations and states}
As usual in logics, the first step in providing semantics to a logic is to provide carrier sets for interpreting individuals, and as a consequence of this, for the symbols in the signature (i.e., the rigid symbols), while states, are assignments of concrete values from these carrier sets to the flexible symbols, that can vary across the model.

We are interested in providing a notion of model where individuals are interpreted over values that can be denoted by terms. That is how we guarantee that such values can effectively be constructed through the operations declared in the signature of the systems under consideration. If $\<\mathsf{Sign}^{\sf L}, \mathbf{Sen}^{\sf L}\>$ is a language and $\Sigma \in |\mathsf{Sign}^{\sf L}|$ then, we would like rigid symbols to be interpreted as constants, functions and relations, respectively, over the carrier set of the term algebra $\mathbf{T}_{\mathbf{RSign}(\Sigma)}(\emptyset)$ (see \cite[Chap.~II, \S10, Def.~10.4]{burris81}) (i.e., $\mathit{Term} (\mathbf{RSign}(\Sigma))$, as defined in Def.~\ref{eq}). From a proof theoretical standpoint, an \emph{interpretation} will be an equational theory presentation over the signature $\mathbf{RSign}(\Sigma)$ providing an axiomatisation of rigid symbols in terms of the equality (i.e., when terms can be considered to be equal, thus being interchangeable within any syntactic construction). 

\begin{definition}[Interpretations]
\label{def:interpretations}
Let ${\sf L} = \<\mathsf{Sign}^{\sf L}, \mathbf{Sen}^{\sf L}\>$ be a language and $\Sigma \in |\mathsf{Sign}^{\sf L}|$, and ${\sf Eq} = \<\mathsf{Sign}^{\sf Eq}, \mathbf{Sen}^{\sf Eq}\>$ the language of equational logic, then we define the functor $\mathbf{Ints}: \mathsf{Sign}^{\sf L} \to \mathsf{Th}^{\sf Eq}$ as follows: 
\begin{itemize}[$-$]
\item let $\Sigma \in |\mathsf{Sign}^{\sf L}|$ then $\mathbf{Ints}(\Sigma)$ is the collection $\setof{\<\mathbf{RSign} (\Sigma), \Gamma\>}{\Gamma \subseteq \mathbf{Sen}^{\sf{Eq}} \circ \mathbf{RSign} (\Sigma)}$ (notice that equational theory presentations are given over ground terms as Def.~\ref{eq} do not consider the availability of flexible symbols), and
\item let $\sigma: \Sigma \to \Sigma' \in ||\mathsf{Sign}^{\sf L}||$ then \\
$\mathbf{Ints}(\sigma): \mathsf{Th}^{\sf Eq}_{\mathbf{RSign} (\Sigma)} \to \mathsf{Th}^{\sf Eq}_{\mathbf{RSign} (\Sigma')}$\footnote{Given $\Sigma \in |\mathsf{Sign}^{\sf Eq}|$, $\mathsf{Th}^{\sf Eq}_{\Sigma}$ denotes the full subcategory of $\mathsf{Th}^{\sf Eq}$ with signature $\Sigma$.} is a functor determined by $\widehat{\sigma}: \mathbf{RSign} (\Sigma) \to \mathbf{RSign} (\Sigma') \in ||\mathsf{Sign}^{\sf Eq}||$ such that: 
\begin{itemize}[$-$]
\item $\mathbf{Ints}(\sigma) (\<\mathbf{RSign}(\Sigma), \Gamma\>) = \<\mathbf{RSign}(\Sigma'), \Gamma'\>$ implies $\mathbf{Sen}^{\sf Eq} (\widehat{\sigma}) (\Gamma) \subseteq \Gamma'$, and
\item for all $\sigma_T: \<\Sigma, \Gamma_1\> \to \<\Sigma, \Gamma_2\> \in ||\mathbf{Ints}(\Sigma)||$, $\mathbf{Ints}(\sigma)(\sigma_T)$ is the unique morphism $\sigma_T'$ such that $\mathbf{Ints} (\sigma) \circ \sigma_T = \sigma'_T \circ \mathbf{Ints} (\sigma)$.
\end{itemize}
\end{itemize}
\end{definition}

Next, we introduce the notion of \emph{definition} as a way of fixing the value of a symbols of a signature, as a term over a different signature.

\begin{definition}[Definitions]
Let ${\sf Eq} = \<\mathsf{Sign}^{\sf Eq}, \mathbf{Sen}^{\sf Eq}\>$ the language of equational logic, we define the bifunctor $\mathbf{Defs}: \mathsf{Sign}^{\sf Eq} \times \mathsf{Sign}^{\sf Eq} \to \mathsf{Set}$ as follows: let $\Sigma = \<C, F, P\>, \Sigma' = \<C', F', P'\> \in |\mathsf{Sign}^{\sf Eq}|$, then $\mathbf{Defs} (\Sigma, \Sigma')$ is the set:
\[
\begin{array}{l}
\left\{c = t \ {\Big |}\  c \in C\ \mathit{and}\ t \in \mathit{Term} (\Sigma')\right\} \cup\\
\left\{f(t_1, \ldots, t_{\mathit{ar}(f)})=t \ {\Big| }\ f \in F\ \text{ and } t_1, \ldots, t_{\mathit{ar}(f)}, t \in \mathit{Term} (\Sigma') \right\} \cup\\
\left\{p(t_1, \ldots, t_{\mathit{ar}(P)}) \ {\Big |}\  p \in P\ \text{ and } t_1, \ldots, t_{\mathit{ar}(f)} \in \mathit{Term} (\Sigma')\right\}
\end{array}
\]
Let $\Sigma_1 = \<C_1, F_1, P_1\>$, $\Sigma'_1 = \<C_1', F_1', P_1'\>$, $\Sigma_2 = \<C_2, F_2, P_2\>$, $\Sigma'_2 = \<C_2', F_2', P_2'\>$ $\in |\mathsf{Sign}^{\sf Eq}|$ be equational signatures, $\sigma_1:  \Sigma_1 \to \Sigma'_1$, $\sigma_2: \Sigma_2 \to \Sigma'_2 \in ||\mathsf{Sign}^{\sf Eq}||$ be equational signature morphisms and $\alpha \in \mathbf{Defs}(\Sigma_1, \Sigma_2)$, then $\mathbf{Defs} (\sigma_1, \sigma_2)$ is defined as follows:
\[
\begin{array}{l}
\mathbf{Defs} (\sigma_1, \sigma_2)(c = t) = \sigma_1(c) = \sigma^*_2 (t)\\
\mathbf{Defs} (\sigma_1, \sigma_2)(f(t_1, \ldots, t_{\mathit{ar}(f)}) = t) =  \sigma_1(f) (\sigma^*_2 (t_1), \ldots, \sigma^*_2 (t_{\mathit{ar}(f)})) = \sigma^*_2 (t)\\
\mathbf{Defs} (\sigma_1, \sigma_2)(p(t_1, \ldots, t_{\mathit{ar}(p)})) =  \sigma_1(p) (\sigma^*_2 (t_1), \ldots, \sigma^*_2 (t_{\mathit{ar}(p)}))
\end{array}
\]
\noindent where, $\sigma^*: \mathbf{Sen}^{\sf Eq}(\Sigma) \to \mathbf{Sen}^{\sf Eq}(\Sigma')$ is the homomorphic extension of $\sigma:\Sigma \to \Sigma' \in ||\mathsf{Sign}^{\sf Eq}||$, to the structure of the formula.
\end{definition}

The intuition behind the previous definition is that a symbol in the signature passed as the first argument is assigned a term over the signature passed as the second argument. This is done by considering three different kinds of formulae:
\begin{inparaenum}[1.]
\item a constant symbol from the first signature is to be interpreted as a term over the second signature,
\item a function symbol from the first signature is defined by the values (a term over the second signature) it yields when it is applied to, as many terms, over the second signature, as the arity of that function symbol prescribes, and,
\item analogously, a predicate symbol from the first signature is defined by the tuples of terms, over the second signature, of the size that the arity of that predicate symbol prescribes, over which it holds.
\end{inparaenum}

Following this approach, \emph{states} are definitions of flexible symbols as terms over rigid symbols. Note that this is consistent with what is done in the traditional model theory where carrier sets, interpreting domains, are fixed by the interpretation forcing flexible symbols (i.e., variables) to take values from those carrier sets.

\begin{definition}[States]
\label{def:states}
States are characterised by the bifunctor $\mathbf{Sts}: \mathsf{Sign}^{\sf Eq} \times \mathsf{Sign}^{\sf Eq} \to \mathsf{Cat}$, defined as: 
\begin{itemize}[$-$]
\item let $\Sigma, \Sigma' \in |\mathsf{Sign}^{\sf Eq}|$ then $\mathbf{Sts}(\Sigma, \Sigma')$ is the collection $\{\<\Sigma + \Sigma', \Gamma\>\ |\ \Gamma \subseteq \mathbf{Defs}(\Sigma, \Sigma')\}$, and 
\item let $\sigma: \Sigma_1 \to \Sigma_2, \sigma': \Sigma'_1 \to \Sigma'_2 \in ||\mathsf{Sign}^{\sf Eq}||$ and $\sigma_T: \<\Sigma_1 + \Sigma'_1, \Gamma\> \to \<\Sigma_1 + \Sigma'_1, \Gamma'\> \in ||\mathbf{Sts}(\Sigma)||$ then \\
$\mathbf{Sts}(\sigma + \sigma')(\<\Sigma_1 + \Sigma'_1, \Gamma\>) = \<\Sigma_2 + \Sigma'_2, \mathbf{Sen}^{\sf Eq}(\sigma_1+\sigma_2)(\Gamma)\>$ and $\mathbf{Sts}(\sigma_1+\sigma_2)(\sigma_T)$ is the unique morphism $\sigma_T'$ such that $\<\sigma + \sigma', [\sigma, \sigma']\> \circ \sigma_T = \sigma'_T \circ \<\sigma + \sigma', [\sigma, \sigma']\>$,
where $[\sigma, \sigma']: \mathbf{Sen}^{\sf Eq}(\Sigma_1 + \Sigma'_1) \to \mathbf{Sen}^{\sf Eq}(\Sigma_2 + \Sigma'_2)$ is defined as follows:
\[
\begin{array}{l}
{[}\sigma, \sigma'{]} (\mathit{in}_l (c) = t) = \mathit{in}_l \circ \sigma (c) = {\sigma'}^* (t)\\
{[}\sigma, \sigma'{]} (p (t_1, \ldots t_n)) = \mathit{in}_l \circ \sigma (p) ({\sigma'}^* (t_1), \ldots, {\sigma'}^* (t_{\mathit{ar}(p)}))\\
{[}\sigma, \sigma'{]} (f (t_1, \ldots t_n) = t) = \mathit{in}_l \circ \sigma (f) ({\sigma'}^* (t_1), \ldots, {\sigma'}^* (t_{\mathit{ar}(p)})) = {\sigma'}^* (t) \\
\ \\
{\sigma'}^* (\mathit{in}_r (c)) = \mathit{in}_r \circ \sigma' (c)\\
{\sigma'}^* (f (t_1, \ldots t_n)) = \mathit{in}_r \circ \sigma' (f) ({\sigma'}^* (t_1), \ldots, {\sigma'}^* (t_{\mathit{ar}(p)}))
\end{array}
\]
\end{itemize}
\end{definition}

Let ${\sf L}$ be a logic, then the functor characterising the category of states for ${\sf L}$ is $\mathbf{States}^{\sf L} = \mathbf{Sts} \circ (\mathbf{RSign}\times\mathbf{FSign})$.

\begin{example}[Interpretations and states]
Let ${\sf L} = \<\mathsf{Sign}^{\sf L}, \mathbf{Sen}^{\sf L}\>$ be a language and $\Sigma \in |\mathsf{Sign}^{\sf L}|$ such that $\mathbf{FSign} (\Sigma) = \Sigma_1$ and $\mathbf{RSign} (\Sigma) = \Sigma_2$

$\Sigma_1 = \<\{x, y\}, \emptyset, \emptyset\>, \Sigma_2 = \<\{0, 1\}, \{+, \cdot\}, \{<\}\> \in |\mathsf{Sign}^{\sf Eq}|$ then, $I \in |\mathbf{Ints}(\Sigma)|$ and $S \in |\mathbf{Sts}(\Sigma_1, \Sigma_2)|$.
$$
\begin{array}{l}
\begin{array}{clc}
& I =\\
\langle & \Sigma_2, \\ 
           & \left\{\begin{array}{c}0 + t = t,\\ 1 \cdot t = t\end{array}\right\}_{t \in \mathit{Term}(\Sigma)} \cup\\
           & \left\{\begin{array}{c}t + t' = t' + t,\\ t \cdot t' = t' \cdot t\end{array}\right\}_{t, t' \in \mathit{Term}(\Sigma)} \cup\\
            & \left\{\begin{array}{c}t + (t' \cdot t'') = (t + t') \cdot (t + t''),\\ t \cdot (t' + t'') = (t \cdot t') + (t \cdot t'')\end{array}\right\}_{t, t', t'' \in \mathit{Term} (\Sigma)} \cup\\
           & \{0 < t\}_{t \in \mathit{Term}(\Sigma), t \not= 0} \cup\\
           & \{t < t + t'\}_{t, t' \in \mathit{Term}(\Sigma), t' \not= 0} & \rangle
\end{array}
\\
\ \\
\begin{array}{clc}
& S =\\
\langle & \Sigma_1 + \Sigma_2, \\ 
            & \{\mathit{in}_l (x) = \mathit{in}_r (0) \mathit{in}_r (+) \mathit{in}_r (1),\\
            & \ \ \mathit{in}_l (y) = \mathit{in}_r (0) \mathit{in}_r (+) \mathit{in}_r (1) \mathit{in}_r (+) \mathit{in}_r (1)\} & \rangle
\end{array}
\end{array}
$$
The reader should note that the equational theory $I$ is expressed in a compact way for the sake of presentation, but it is an infinite set of formulae.
\end{example}

\subsection{On the relational structure of models}
Algebraisations of relational models have been studied for a long time; examples of this are the study of modal logic from a relation algebraic perspective  by Schlingloff and Heinle in \cite{schlingloff:brink97} and the many interpretability results of a variety of modal, and multimodal, logics \cite{frias:jancl-8,frias:relmics01+,frias:relmics03,frias:jlap-66_2} in $\omega$-closure fork algebras with urelements \cite[Def.~7]{frias:relmics01+}.
While all these works share the common interest of building an algebraic framework in which it is possible to embed logical systems of a modal nature, they also share the feature that such an embedding is obtained in an \emph{ad hoc} manner by finding workarounds for representing their syntactic features (i.e., terms, formulae, etc.) as equations and their semantics (i.e., models) as algebras. In this work we dig into this previous experience in order to rescue a more general methodology for obtaining proof-theoretically supported classes of models.

The first step in providing a representation for a class of relational models is to formally determine the class of relational structures over which such models will be defined. Therefore, we will first concentrate on building a formalisation of Kripke structures; it is immediate to see that a labelled transition systems over a given set $S$ (see Def.~\ref{def:lts}) can be thought of as the interpretation of a specific $\mathsf{ETR*}$-signature, over the full proper closure relation algebras with base set $S$. On the other hand, if we assume that Kripke structures are labelled transition systems satisfying certain properties in order to fit the needs of the logic being interpreted, then we can consider using $\mathsf{ETR*}$-theory presentations for formalising classes of Kripke structures, and the calculus for $\mathsf{ETR*}$ as a formal tool for reasoning about their properties.

To accomplish this, there is a gap to bridge between the intended semantics for a logic and the actual class of models determined by the $\mathsf{ETR*}$ formalisation of such an intuitive notion of a relational model. More concretely, given a logic $\mathsf{L}$, for which we have a functor $\mathbf{RelTh}: \mathsf{Sign}^{\sf L} \to \mathsf{Th}^{\sf ETR*}$ mapping its signatures to $\mathsf{ETR*}$-theory presentations, we need to prove that, given $\Sigma \in |\mathsf{Sign}^\mathsf{L}|$, the class of models in $\mathbf{Mod}^\mathsf{ETR*} (\mathbf{RelTh} (\Sigma))$ conforms to the intuition we have of the semantics for $\mathsf{L}$, meaning, that they are (a specific class of) labelled transition systems over what we consider states of the system (i.e., objects in $|\mathbf{States}^{\sf L} (\Sigma)|$). 

Let $\mathsf{L}$ be a logic with language $\<\mathsf{Sign}^{\sf L}, \mathbf{Sen}^{\sf L}\>$ and category of theories $\mathsf{Th}^{\sf L}$, $\mathbf{Sign}^{\sf L}: \mathsf{Th}^{\sf L} \to \mathsf{Sign}^{\sf L}$ is the forgetful functor such that: if $\<\Sigma, \Gamma\> \in |\mathsf{Th}^{\sf L}|$ then $\mathbf{Sign}^{\sf L} (\<\Sigma, \Gamma\>) = \Sigma$ and if $\sigma: \<\Sigma, \Gamma\> \to \<\Sigma', \Gamma'\> \in ||\mathsf{Th}^{\sf L}||$ then $\mathbf{Sign}^{\sf L} (\sigma) = \sigma$, the underlying signature morphism over which the theory morphism is defined. When it is clear from the context, the superscripting with ${\sf L}$ will be omitted.

The next definition provides the characterisation of the category of models of interest for the purpose of formalising relational models.

\begin{definition}
\label{spcra}
Let $\mathsf{L}$ be a logic, $\mathbf{RelTh}: \mathsf{Sign}^{\sf L} \to \mathsf{Th}^{\sf ETR*}$ a functor and $\Sigma \in |\mathsf{Sign}^\mathsf{L}|$, $\mathcal{A} \in |\mathbf{Mod}^\mathsf{ETR*} \circ \mathbf{RelTh} (\Sigma)|$ is said to be a \emph{state proper closure relation algebra} if its base set is a subset of $|\mathbf{States}^{\sf L} (\Sigma)|$. Given a state proper closure relation algebra $\mathcal{M}$, we will refer to its base set as $\mathcal{M}_\mathit{bs}$.

Let $T = \mathbf{RelTh} (\Sigma)$, we define $\mathbf{Mod}^\mathsf{sETR*} (T)$ as $\<\mathcal{O}, \mathcal{A}\>$ where:
\begin{itemize}[$-$]
\item $\mathcal{O}$ is the class of the state proper closure relation algebras, and 
\item $\mathcal{A}$ are those state proper closure relation algebras homomorphism $\gamma_h: \mathcal{M} \to \mathcal{M}' \in ||\mathbf{Mod}^\mathsf{ETR*} (\mathbf{RelTh} (\Sigma))||$, with $h: \mathcal{M}_\mathit{bs} \to {\mathcal{M}'}_\mathit{bs}$, satisfying: 
 \begin{itemize}[$-$]
 \item for all $s \in |\mathcal{M}|$, $s^\bullet = \(h (s)\)^\bullet$,
 \item {\bf forward condition:} for all $R \in |\mathcal{M}|$, if $\<s_1, s_2\> \in R$ then $\<h (s_1), h (s_2)\> \in \gamma (R)$, and
 \item {\bf backward condition:} for all $R \in |\mathcal{M}|$, if $\<h (s_1), s'_2\> \in \gamma (R)$ then there exists $s_2 \in \mathcal{M}_\mathit{bs}$ such that $h (s_2) = s'_2$ and $\<s_1, s_2\> \in R$.
\end{itemize}
\end{itemize}
\end{definition}

Definition~\ref{spcra} provides a relational characterisation of what is generally known as a \emph{Kripke frame}, the cornerstone of relational semantics (see, for instance, \cite[Def.~1.19]{blackburn01} for the case of basic modal logic or \cite[Sec.~5.2]{harel00} for propositional dynamic logic; in the cases of computational tree logic \cite[pp.~166]{benari:acm-sigplan-sigact81} and linear temporal logic \cite[pp.~52]{pnueli:tcs-13_1}, there is an implicit use of this notion of frame for defining trees and sequences of states of a system, respectively). The interesting part of the previous definition is that, as we are interested in algebraically characterising Kripke frames, not every state proper closure relation algebra homomorphism is considered to be a relational structure morphism, but those representing \emph{bounded morphism} \cite[Def.~2.12]{blackburn01}, the natural notion of morphism between such structures.

\begin{lemma}
Let $\mathsf{L}$ be a logic, $\mathbf{RelTh}: \mathsf{Sign}^{\sf L} \to \mathsf{Th}^{\sf ETR*}$ a functor and $\Sigma \in |\mathsf{Sign}^\mathsf{L}|$, if $T = \mathbf{RelTh} (\Sigma)$ then $\mathbf{Mod}^\mathsf{sETR*} (T)$ is a category.
\end{lemma}
\begin{proof}
The proof follows straightforwardly by observing that: 
\begin{inparaenum}[1)]
\item for every state proper closure relation algebra there exists an identity homomorphism and an identity function over its base set,
\item composition of algebra homomorphisms satisfying the conditions of Def.~\ref{spcra}, is an homomorphism, also satisfy those conditions, 
\item identity homomorphism behave as identities when they are composed, and 
\item composition of homomorphism from that class is associative.
\end{inparaenum}
\end{proof}

\begin{corollary}
Let $\mathsf{L}$ be a logic, $\mathbf{RelTh}: \mathsf{Sign}^{\sf L} \to \mathsf{Th}^{\sf ETR*}$ a functor and $\Sigma \in |\mathsf{Sign}^\mathsf{L}|$, if $T = \mathbf{RelTh} (\Sigma)$ then $\mathbf{Mod}^\mathsf{sETR*} (T) \subseteq \mathbf{Mod}^\mathsf{ETR*} (T)$
\qed
\end{corollary}

The satisfaction relation $\models^\mathsf{sETR*}_{\mathbf{Sign} (T)} \subseteq |\mathbf{Mod}^\mathsf{sETR*} (T)| \times \mathbf{Sen}^\mathsf{ETR*} \circ \mathbf{Sign} (T)$ is the restriction of $\models^\mathsf{ETR*}_{\mathbf{Sign} (T)}$ to algebras in $\mathbf{Mod}^\mathsf{sETR*} (T)$.

The following fact about $\mathsf{ETR*}$-models follows trivially by observing that $\mathsf{ETR*}$-formulae cannot prescribe specific properties of the base set of proper closure relation algebras.

Then, after these results, we can only expect that, for every proper closure relation algebras satisfying the algebraic axiomatisation of the class of Kripke models for a logic $\mathsf{L}$, there exists a state proper closure relation algebras satisfying the same formulae. This is usually referred to as the representability of the class of proper closure relation algebras in the class of state proper closure relation algebras.

\begin{lemma}\ \\
\label{lemma:complete}
Let $\mathsf{L}$ be a logic, $\mathbf{RelTh}: \mathsf{Sign}^{\sf L} \to \mathsf{Th}^{\sf ETR*}$ a functor and $\Sigma \in |\mathsf{Sign}^\mathsf{L}|$, then for all $\mathcal{M} \in |\mathbf{Mod}^\mathsf{ETR*} (\mathbf{RelTh} (\Sigma))|$ there exists $\mathcal{M}' \in |\mathbf{Mod}^\mathsf{sETR*} (\mathbf{RelTh} (\Sigma))|$ such that for all $\alpha \in \mathbf{Sen}^\mathsf{ETR*} \circ \mathbf{Sign} (\mathbf{RelTh} (\Sigma))$, if $\mathcal{M} \models^\mathsf{ETR*}_{\mathbf{Sign} (\mathbf{RelTh} (\Sigma))} \alpha$ then $\mathcal{M}' \models^\mathsf{sETR*}_{\mathbf{Sign} (\mathbf{RelTh} (\Sigma))} \alpha$.
\end{lemma}
\begin{proof}
By Coro.~\ref{coro-skolem-lowenheim}, if there exists a model $\mathcal{M} \in |\mathbf{Mod}^\mathsf{ETR*} (\mathbf{RelTh} (\Sigma))|$, then there exists $\mathcal{M}_1 \in |\mathbf{Mod}^\mathsf{ETR*} (\mathbf{RelTh} (\Sigma))|$ whose underlying proper closure relation algebra has a base set $U$ such that $|U| \leq \omega_1$. Then, by Coro~\ref{coro:full}, there exists $\mathcal{M}_2 \in |\mathbf{Mod}^\mathsf{ETR*} (\mathbf{RelTh} (\Sigma))|$ such that:
\begin{inparaenum}[1.]
\item its underlying proper closure relation algebra is full and has the same base set $U$, and
\item for all $\alpha \in \mathbf{Sen}^\mathsf{ETR*} \circ \mathbf{Sign} (\mathbf{RelTh} (\Sigma))$, $\mathcal{M}_1 \models^\mathsf{ETR*}_{\mathbf{Sign} (\mathbf{RelTh} (\Sigma))} \alpha$ implies $\mathcal{M}_2 \models^\mathsf{ETR*}_{\mathbf{Sign} (\mathbf{RelTh} (\Sigma))} \alpha$.
\end{inparaenum}

Finally, by Prop.~\ref{prop:uniquefull}, $\mathcal{M}_2$ is unique up-to isomorphism so we can consider $\mathcal{M}' \in |\mathbf{Mod}^\mathsf{ETR*} (\mathbf{RelTh} (\Sigma))|$, isomorphic to $\mathcal{M}_2$, but whose underlying full proper closure relation algebra has as its base set a subset of $|\mathbf{States}^{\sf L} (\Sigma)|$. Thus, $\mathcal{M}' \in |\mathbf{Mod}^\mathsf{sETR*} (\Sigma)|$.
\end{proof}

\begin{theorem}
\label{thm:complete}
Let $\mathsf{L}$ be a logic, $\mathbf{RelTh}: \mathsf{Sign}^{\sf L} \to \mathsf{Th}^{\sf ETR*}$ a functor and $\Sigma \in |\mathsf{Sign}^\mathsf{L}|$, then:
\begin{itemize}[$-$]
\item \textbf{Sound: } $\mathbf{RelTh} (\Sigma) \vdash^\mathsf{ETR*}_{\mathbf{Sign} (\mathbf{RelTh} (\Sigma))} \alpha$ implies \\
$\mathbf{RelTh} (\Sigma) \models^\mathsf{sETR*}_{\mathbf{Sign} (\mathbf{RelTh} (\Sigma))} \alpha$
\item \textbf{Complete: } $\mathbf{RelTh} (\Sigma) \models^\mathsf{sETR*}_{\mathbf{Sign} (\mathbf{RelTh} (\Sigma))} \alpha$ implies \\
$\mathbf{RelTh} (\Sigma) \vdash^\mathsf{ETR*}_{\mathbf{Sign} (\mathbf{RelTh} (\Sigma))} \alpha$
\end{itemize}
\end{theorem}
\begin{proof}
The \emph{soundness} part of the theorem follows trivially by observing that $|\mathbf{Mod}^\mathsf{sETR*} (\mathbf{RelTh} (\Sigma))| \subset |\mathbf{Mod}^\mathsf{ETR*} (\mathbf{RelTh} (\Sigma))|$.

To prove \emph{completeness} lets assume that $\alpha$ is not provable from $\mathbf{RelTh} (\Sigma)$ (i.e., $\mathbf{RelTh} (\Sigma) \vdash^\mathsf{ETR*}_{\mathbf{Sign} (\mathbf{RelTh} (\Sigma))} \alpha$). Then, there exists $\mathcal{M} \in |\mathbf{Mod}^\mathsf{ETR*} (\mathbf{RelTh} (\Sigma))|$ such that $\mathcal{M} \models^\mathsf{ETR*}_{\mathbf{Sign} (\mathbf{RelTh} (\Sigma))} \alpha$ does not hold. Therefore, by Def.~\ref{def:etr*-satisfaction}, $\mathcal{M} \models^\mathsf{ETR*}_{\mathbf{Sign} (\mathbf{RelTh} (\Sigma))} \neg\alpha$ and, by Lemma~\ref{lemma:complete}, there exists $\mathcal{M}' \in |\mathbf{Mod}^\mathsf{sETR*} (\mathbf{RelTh} (\Sigma))|$ such that $\mathcal{M}' \models^\mathsf{sETR*}_{\mathbf{Sign} (\mathbf{RelTh} (\Sigma))} \neg\alpha$ and, consequently, $\mathcal{M}' \models^\mathsf{sETR*}_{\mathbf{Sign} (\mathbf{RelTh} (\Sigma))} \alpha$ does not hold and, consequently, $\mathbf{RelTh}(\Sigma) \models^\mathsf{sETR*}_{\mathbf{Sign} (\mathbf{RelTh} (\Sigma))} \alpha$ does not hold, either.
\end{proof}

The previous theorem guaranties that given a logic $\mathsf{L}$ and a functor $\mathbf{RelTh}: \mathsf{Sign}^{\sf L} \to \mathsf{Th}^{\sf ETR*}$, for all $\Sigma \in |\mathsf{Sign}^\mathsf{L}|$, the entailment relation $\vdash^\mathsf{ETR*}_{\mathbf{Sign} (\mathbf{RelTh} (\Sigma))} \subseteq 2^{\mathbf{Sen}^\mathsf{ETR*} \circ \mathbf{Sign} (\mathbf{RelTh} (\Sigma))} \times \mathbf{Sen}^\mathsf{ETR*} \circ \mathbf{Sign} (\mathbf{RelTh} (\Sigma))$ provides a sound and complete calculus for reasoning about the formal properties of the class of Kripke models over states from $\mathbf{States}^\mathsf{L} (\Sigma)$. 

Let ${\sf L}$ be a logic and $\mathbf{RelTh}: \mathsf{Sign}^{\sf L} \to \mathsf{Th}^{\sf ETR*}$ a functor, then the functor characterising the category of relational structures for ${\sf L}$ is $\mathbf{RelStr}^{\sf L} = \mathbf{Mod}^\mathsf{sETR*} \circ \mathbf{RelTh}$.

\begin{example}[Relational structure of first order dynamic logic models]
\label{relth-fodl}
Let $\<\mathsf{Sign}^{\sf FODL}, \mathbf{Sen}^{\sf FODL}\>$ the language of first order dynamic logic, $\Sigma = \<C, F, P, A\> \in |\mathsf{Sign}^{\sf FODL}|$ and $X$ be a countable set of individual variable symbols. Let us assume that atomic programs, denoted by symbols in $\Sigma$, are required to be total and functional (i.e., deterministic); then we define $\mathbf{RelTh}(\Sigma) = \<\Sigma^{\mathit{Rel}}, \Gamma^{\mathit{Rel}}\>$ where:
$$
\begin{array}{rcl}
\Sigma^{\mathit{Rel}} & = & \<\{A_a\}_{a \in A}\>\\
\Gamma^{\mathit{Rel}} & =  & \{(\forall x,y)(x\ (A_a\acompo\aconv{A_a})\ameet\aid\ y \iff x\ \aid\ y)\}_{a \in A} \cup\\
                                     &     & \qquad \mbox{[Atomic programs are total]}\\
                                     &     & \{(\forall x, y)(x\ A_a;\aconv{A_a} \ y \implies x\ \aid\ y)\}_{a \in A} \\
                                     &     & \qquad \mbox{[Atomic programs are functional (i.e. deterministic) ]}
                                     \end{array}
$$
Notice that further restriction on the behaviour of atomic actions, like having a precondition and a postcondition, must be specified by means of specific formulae of the form $\mathit{pre}_a \implies [A_a]\mathit{post}_a$, where $a \in A$, and it is the satisfaction relation the key element in ensuring the appropriateness of the relational interpretations.
\end{example}

Finally, the key element in completing the definition of our framework is the definition of the class of models associated to a relational structure. Let $\mathsf{L} = \<\mathsf{Sign}^{\sf L}, \mathbf{Sen}^{\sf L}\>$ be a language with state sublanguage $\<\mathbf{StSen}, \mathbf{RSign}, \mathbf{FSign}, \rho^{Sen}\>$ and whose class of relational structures is given by the functor $\mathbf{RelTh}: \mathsf{Sign}^{\sf L} \to \mathsf{Th}^{\sf ETR*}$ then, models for a signature $\Sigma \in |\mathsf{Sign}^{\sf L}|$ are defined as a category of structures $\mathsf{Struct}^{\sf L}_\Sigma$. In general, such a class of structures acting as models for ${\sf L}$-formulae will be closely related to the relational structures identified by the functor $\mathbf{RelStr}$ and, in most cases, it requires some form of structuring of the states over which the relational structure is defined. Among the plethora of logical formalisms for software specification, some prominent examples for exposing such a diversity of notions of model are:
\begin{inparaenum}[a)]
\item \emph{Linear Temporal Logic}, whose semantics is given in terms of infinite sequences of states, reflecting linear paths in the relational structure (see \cite[pp.~52]{pnueli:tcs-13_1}),
\item \emph{Temporal Logic}, whose semantics is given in terms of pairs consisting of an infinite sequence of states and a natural number, reflecting a linear path in the relational structure and a pointer to a specific position in it (see \cite[Sec~3.3]{manna91}),
\item \emph{Computational Tree Logic} and \emph{Dynamic Logic}, whose semantics is defined over a specific state of the relational structure (see \cite[pp.~166]{benari:acm-sigplan-sigact81} and  \cite[pp.~287]{fischer:stoc77}, respectively), and 
\item \emph{Computational Tree Logic}$*$, whose semantics is given, heterogeneously, in terms of states and linear paths in the relational structure, depending on whether the specific formula is a state formula, or a path formula (see \cite[Sec.~2.3]{emerson:jcss-30_1}).
\end{inparaenum}

Such a diversity prevents us from attempting a confinement of the notion of model, sending us in a direction similar to the one chosen by Meseguer in his formalisation of \emph{Proof Calculus} \cite[Def.~12]{meseguer:lc87}, where proofs are only required to be some form of structure organised as a category. In our case, this last requirement, will be specially useful because model morphisms are of utmost importance for enabling compositional semantics, where the semantics of composite specifications is given as a combination of the semantics of the compounds.

In many cases, like the ones mentioned above, such structures are containers of states whose specific properties can be expressed as \emph{generic} datatypes (see \cite[Sec.~2.4]{backhouse:afp98} for a lightweight introduction to the topic or \cite{fokkinga:phdthesis,fokkinga:mscs-6_1,hoogendijk:phdthesis} for complete presentation including all the formalities). In just a few words, a generic datatype is the fixpoint of a pattern functor that relates two categories in which the objects (resp. morphisms) of the target one are arrangements (according to the structure imposed by the functor) of objects (resp. morphisms) of the source one.

\begin{example}[Models for first order dynamic logic]
\label{relstr-fodl}
Let $\<\mathsf{Sign}^{\sf FODL}, \mathbf{Sen}^{\sf FODL}\>$ be the language of ${\sf FODL}$, $\Sigma \in |\mathsf{Sign}^{\sf FODL}|$, $X$ be a countable set of individual variable symbols and $\mathbf{RelTh}(\Sigma) = \<\Sigma^{\mathit{Rel}}, \Gamma^{\mathit{Rel}}\>$ defined as in Ex.~\ref{relth-fodl}. 

As we shown Def.~\ref{spcra}, $\mathbf{RelStr} (\Sigma)$ is the category whose objects are those state proper closure relation algebras over the signature $\Sigma^{\mathit{Rel}}$ satisfying $\Gamma^{\mathit{Rel}}$ (representing Kripke frames), and whose morphisms are certain proper closure relation algebras homomorphisms between state proper closure relation algebras (representing bounded morphisms between Kripke frames). Finally, $\mathsf{Struct}^{\sf FODL}_\Sigma = \<\mathcal{O}, \mathcal{A}\>$ where:
\begin{itemize}[$-$]
\item $\mathcal{O} = \left\{\<I, \mathcal{M}, s\> \ {\Big |}\ \in |\mathbf{Ints} (\Sigma)|, \mathcal{M} \in |\mathbf{RelStr} (\Sigma)|,  s \in |\mathbf{States} (\Sigma)|_{\mathcal{M}_\mathit{bs}}|\right\}$, and

\item $\mathcal{A} = \left\{\<\sigma, \gamma_h, h\>: \<I, \mathcal{M}, s\> \to \<I', \mathcal{M}', s'\>\ |\ \sigma: I \to I' \in ||\mathbf{Ints} (\Sigma)||, \right.$

$\left. \gamma_h:\mathcal{M} \to \mathcal{M}' \in ||\mathbf{RelStr} (\Sigma)||, h: s \to s' \in ||\mathbf{States} (\Sigma)|_{\mathcal{M}_\mathit{bs}}|| \right\} $
\end{itemize}
Note that we are implicitly saying that $h$ is a function between the base sets of $\mathcal{M}$ and $\mathcal{M}'$, respectively, satisfying the conditions of Def.~\ref{spcra}.
\end{example}

Disregarding the specific features of a logic, by the way in which it's language was defined (i.e., by means of it's syntax $\<\mathsf{Sign}, \mathbf{Sen}\>$, and it's state sublanguage $\<\mathbf{StSen}, \mathbf{RSign}, \mathbf{FSign}, \rho^{Sen}\>$), the satisfaction relation can be split in two, depending on whether the formula under interpretation is a state formula or not. In the case of state formula, its satisfaction is only subject to the combination of:
\begin{inparaenum}[1)]
\item the interpretation, proving meaning to the rigid symbols. and 
\item some notion of current state of the model, providing meaning to flexible symbols.
\end{inparaenum}

The next definition provides an homogeneous definition of satisfaction for formulae from the state sublanguage of a logic.

\begin{definition}[Satisfaction relation for state formulae]
\label{def:satisfiability}
Let $\<\mathsf{Sign}^{\sf L}, \mathbf{Sen}^{\sf L}\>$ be a language of $\mathsf{L}$, with state sublanguage $\<\mathbf{StSen}, \mathbf{RSign}, \mathbf{FSign}, \rho^{Sen}\>$. Let $\Sigma \in |\mathsf{Sign}^{\sf L}|$ and $\alpha \in \mathbf{StSen} (\Sigma)$, then $\models^{\sf L}_\Sigma \subseteq |\mathbf{Ints} (\Sigma) \times \mathbf{States} (\Sigma)| \times \mathbf{StSen}^{\sf L}(\Sigma)$ is defined as:
\begin{center}
$\<I, s\> \models^{\sf L}_\Sigma \alpha$ iff $T \vdash^{\sf Eq}_{(\mathbf{RSign}+\mathbf{FSign})(\Sigma)} \rho^{\mathit{Sen}}_\Sigma (\alpha)$
\end{center} 
\noindent, where the pair of morphisms $\<\mathit{in}_{l}: I \to T, \mathit{id}_{(\mathbf{RSign}+\mathbf{FSign})(\Sigma)}: s \to T\>$ is a pushout for $\<id_{\mathbf{RSign}(\Sigma)}:\<\mathbf{RSign}(\Sigma), \emptyset\> \to I, \mathit{in}_{l}: \<\mathbf{RSign}(\Sigma), \emptyset\> \to s\>$.
\end{definition}

The following example shows how the previous definition can be extended to obtain the satisfaction relation of first order dynamic logic.

\begin{example}[Satisfaction relation for first order dynamic logic]
\label{ex:satisfaction}
Let $\<\mathsf{Sign}^{\sf FODL}, \mathbf{Sen}^{\sf FODL}\>$ be the language of ${\sf FODL}$, with state sublanguage $\<\mathbf{StSen}, \mathbf{RSign}, \mathbf{FSign}, \rho^{Sen}\>$. Let $\Sigma \in |\mathsf{Sign}^{\sf FODL}|$, $\varphi \in \mathbf{StSen} (\Sigma)$, $\alpha, \beta \in \mathbf{Sen} (\Sigma)$ and $\<I, \mathcal{M}, s\> \in |\mathbf{Mod}^{\sf FODL}(\Sigma)|$, then we define $\models^{\sf FODL}_\Sigma \subseteq \mathbf{Mod}^{\sf FODL}(\Sigma) \times \mathbf{Sen}^{\sf FODL}(\Sigma)$ in the following way:
\[
\begin{array}{l}
\<I, \mathcal{M}, s\> \models^{\sf FODL}_\Sigma \varphi \mbox{ iff } \<I, s\> \models^{\sf FODL}_\Sigma \varphi \\
\<I, \mathcal{M}, s\> \models^{\sf FODL}_\Sigma \neg\alpha \mbox{ iff } \<I, \mathcal{M}, s\> \models^{\sf FODL}_\Sigma \alpha \mbox{ does not hold} \\
\<I, \mathcal{M}, s\> \models^{\sf FODL}_\Sigma \alpha \lor \beta \mbox{ iff } \<I, \mathcal{M}, s\> \models^{\sf FODL}_\Sigma \alpha \mbox{ or } = \<I, \mathcal{M}, s\> \models^{\sf FODL}_\Sigma \beta \\
\<I, \mathcal{M}, s\> \models^{\sf FODL}_\Sigma (\exists x)\alpha \mbox{ iff } \mbox{there exists $s' \in |\mathbf{States} (\Sigma)|_{\mathcal{M}_\mathit{bs}}|$ such that if:}\\
- \mbox{the pair $\<\mathit{in}_{l}: I \to T, \mathit{id}_{(\mathbf{RSign}+\mathbf{FSign})(\Sigma)}: s \to T\>$ is} \\
\mbox{the pushout for $\<id_{\mathbf{RSign}(\Sigma)}:\<\mathbf{RSign}(\Sigma), \emptyset\> \to I,  \mathit{in}_{l}: \<\mathbf{RSign}(\Sigma), \emptyset\> \to s\>$, and}\\
- \mbox{the pair $\<\mathit{in}_{l}: I \to T', \mathit{id}_{(\mathbf{RSign}+\mathbf{FSign})(\Sigma)}: s' \to T'\>$ is} \\
\mbox{the pushout for $\<id_{\mathbf{RSign}(\Sigma)}:\<\mathbf{RSign}(\Sigma), \emptyset\> \to I, \mathit{in}_{l}: \<\mathbf{RSign}(\Sigma), \emptyset\> \to s'\>$, }\\
\mbox{then: } \mbox{for all $\beta \in \mathbf{Sen}^{\sf Eq}(\mathbf{RSign}(\Sigma)+(\mathbf{FSign}(\Sigma)/\{x\}))$,} \\
\mbox{$T \vdash^{\sf Eq}_{(\mathbf{RSign}+\mathbf{FSign})(\Sigma)} \mathbf{Sen}^{\sf Eq}(\varphi)(\beta)$ iff $T' \vdash^{\sf Eq}_{(\mathbf{RSign}+\mathbf{FSign})(\Sigma)} \mathbf{Sen}^{\sf Eq}(\varphi)(\beta)$}\\
\mbox{such that $\varphi$ is the identity over the signature $\mathbf{RSign}(\Sigma)+(\mathbf{FSign}(\Sigma)/\{x\})$, }\\
\mbox{mapped to the signature $\mathbf{RSign}(\Sigma)+\mathbf{FSign}(\Sigma)$, and $\<I, \mathcal{M}, s'\> \models^{\sf FODL}_\Sigma \alpha$ }\\
\<I, \mathcal{M}, s\> \models^{\sf FODL}_\Sigma \<P\>\alpha \mbox{ iff } \mbox{there exists $s' \in |\mathbf{States} (\Sigma)|_{\mathcal{M}_\mathit{bs}}|$ such that: } \\
\hfill  \<s, s'\> \in m_\mathcal{M} (P) \mbox{ and } \<I, \mathcal{M}, s'\> \models^{\sf FODL}_\Sigma \alpha
\end{array}
\]
\noindent where $m_\mathcal{M}$ is defined as follows:
\[
\begin{array}{l}
m_\mathcal{M} (P) = P^\mathcal{M} \mbox{, for all $P \in \Sigma$}\\
m_\mathcal{M} (Q + Q') = m_\mathcal{M} (Q) + m_\mathcal{M} (Q')\\
m_\mathcal{M} (Q ; Q') = m_\mathcal{M} (Q) \circ m_\mathcal{M} (Q')\\
m_\mathcal{M} (\varphi?) = \setof{\<s, s\>}{\<I, \mathcal{M}, s\> \models^{\sf FODL}_\Sigma \varphi}\\
m_\mathcal{M} (Q^*) = \(m_\mathcal{M} (Q)\)^*
\end{array}
\]
\end{example}

\subsection{Relational models for propositional and higher order languages}
\label{sec:hoeq}
As we mentioned before, our main interest is providing a modular approach for characterising relational models over concrete clases of values. In the previous sections we confined ourselves to the case of first order states; this is the reason why we choose equational logic as the language for characterising them. From this point of view, propositional logics such as propositional dynamic logic \cite{harel00}, linear temporal logics \cite{pnueli:tcs-13_1}, computational tree logic \cite{pnueli:ieee-focs77,benari:acm-sigplan-sigact81}, computational tree logic star \cite{emerson:jacm-33_1}, are considered to be $0$-order languages, as they do not incorporate any notion of term; $1$-order languages will be those that have a set of terms interpreted over elements from a certain domain, $2$-order languages will be those that include terms that are interpreted over elements from a domain, but also terms that are interpreted over sets of these elements, that is, $1$-order predicates and so on. This is done in accordance with the type theory presented in \cite[Vol.~I,~Ch.~II]{whitehead27} providing syntax for terms of all finite orders.

Next, we define higher-order equational logic \cite{meinke:tcs-100_2} as an extension of the definition of equational logic given in Def.~\ref{eq}. A detailed presentation of higher-order logic and its logical properties can be found in \cite{vanbenthem:hlfcs83}. Higher-order equational logic, as it is presented in the next definition, can be proved equipolent to \cite[Defs.~1.1,~1.2~and~1.3]{meinke:tcs-100_2} and also to the functional reduct of Van Benthem's presentation of higher-order logic based on structured types for functions \cite{vanbenthem:hlfcs83}.

Let $\mathit{order}: \bigcup_{i=0}^{n} \Sigma_i \times \NAT \to \NAT$ be a functions such that for every $0 \leq i \leq n$ if symbol $f \in \Sigma_i$, and for all $0 \leq j \leq \mathit{ar}(f)$, $\mathit{order}(f, j) = k$ if $k < i$ and we expect the $j$-th. argument of $f$ to be a term in $\mathit{Term}(\Sigma_k)$.

\begin{definition}[Higher-order Equational Logic]
\label{hoeq}
The language of \emph{Higher-order Equational Logic} is a structure $\langle {\sf Sign^{HOEq}}, {\bf Sen}^{\sf HOEq} \rangle$ (denoted ${\sf HOEq}$) such that: 
\begin{itemize}[$-$]
\item ${\sf Sign^{HOEq}} = \langle \mathcal{O}, \mathcal{A} \rangle$ where:
\begin{itemize}[$-$]
\item $\mathcal{O} = \left\{ \bigcup_{i=0}^{n} \Sigma_i \ {\Big |}\ \Sigma_0 = \langle \emptyset, \emptyset, \{P_k^0\}_{k \in K_0} \rangle \right.$

\hfill $\left. \Sigma_n = \langle \{C_j^n\}_{j \in J_n}, \{f^n_i\}_{i \in I_n}, \{P_k^n\}_{k \in K_n} \rangle \right\}$, and
\item $\mathcal{A} = \left\{ \{\sigma_n\}_{n \leq n'}: \bigcup_{i=0}^{n} \Sigma_i \to \bigcup_{i=0}^{n'} \Sigma'_i \ {\Big |}\ \right.$

\hfill $\left. \mbox{for all } 0 \leq i \leq n \in \NAT,\\ \sigma_i: \Sigma_i \to \Sigma'_i \in ||{\sf Sign^{Eq}}|| \right\}$;
\end{itemize}

\item ${\bf Sen}^{\sf HOEq}: \mathsf{Sign^{HOEq}} \to \mathbf{Set}$ is defined as follows:
\begin{itemize}[$-$]
\item let $\Sigma = \bigcup_{i=0}^{n} \Sigma_i \in |{\sf Sign^{HOEq}}|$, such that $\Sigma_i = \<C_i, F_i, P_i\>$ then for all $0 < i \leq n$, $\mathit{Term}(\Sigma_i)$ as the smallest set satisfying:
\begin{enumerate}[1)]
\item for all $1 \leq i \leq n$, $C_i \subseteq \mathit{Term}(\Sigma_i)$, and 
\item if $f \in F_i$ and for all $0 \leq j \leq \mathit{ar}(f)$, $t_j \in \mathit{Term}(\Sigma_{\mathit{order}(f, j)})$ then $f (t_1, \ldots,  t_{\mathit{ar}(f)}) \in \mathit{Term} (\Sigma_i)$, and
\end{enumerate}
${\bf Sen}^{\sf HOEq}(\Sigma)$ is the set:
\[
\begin{array}{l}
\left\{ t = t' \ | \ t, t' \in \mathit{Term}(\Sigma_i) \mbox{, for some } 0 \leq i \leq n \right\} \cup\\
\left\{ P(t_1, \ldots, t_{\mathit{ar}(P)}) \ | \ t_i \in \mathit{Term}(\Sigma_{\mathit{order}(P, i)}) \mbox{, for all } 0 \leq i \leq \mathit{ar}(P) \right\}
\end{array}
\]
\item let $\Sigma = \bigcup_{i=0}^{n} \Sigma_i \in |{\sf Sign^{HOEq}}|$, such that for all $0 \leq i \leq n$, $\Sigma_i = \<C_i, F_i, P_i\>$ and $\Sigma' = \bigcup_{i=0}^{n'} \Sigma'_i \in |{\sf Sign^{HOEq}}|$, such that for all $0 \leq i \leq n'$, $\Sigma'_i = \<C'_i, F'_i, P'_i\>$, if $\sigma = \{\sigma_i\}_{0 \leq i \leq n}: \Sigma \to \Sigma' \in ||{\sf Sign^{HOEq}}||$ such that for all $0 \leq i \leq n$, $\sigma_i = \langle {\sigma_C}_i, {\sigma_F}_i, {\sigma_P}_i \rangle$, then we define $\sigma^*: \mathit{Term}(\Sigma) \to \mathit{Term}(\Sigma')$ as follows:
\begin{itemize}[$-$]
\item for all $0 \leq i \leq n$, $c \in C_i$, $\sigma^* (c) = {\sigma_C}_i (c)$, and
\item for all $0 \leq i \leq n$, $f (t_1, \ldots, t_{\mathit{ar}(f)}) \in Term(\langle C_i, F_i, P_i \rangle)$, 

\hfill $\sigma^* (f (t_1, \ldots, t_{\mathit{ar}(f)})) = {\sigma_F}_i (f) (\sigma^* (t_1), \ldots, \sigma^* (t_{\mathit{ar}(f)}))$.
\end{itemize}
Then, ${\bf Sen}^{\sf HOEq}(\sigma)$ is defined as follows:
\begin{itemize}[$-$]
\item ${\bf Sen}^{\sf HOEq}(\sigma)(t = t') = \sigma^* (t) = \sigma^* (t')$, 
\item ${\bf Sen}^{\sf HOEq}(\sigma)(P (t_1, \ldots, t_{\mathit{ar}(P)})) = {{\sigma_P}_k} (P) (\sigma^* (t_1), \ldots, \sigma^* (t_{\mathit{ar}(P)}))$.
\end{itemize}
\end{itemize}
\end{itemize}
We assume the existence of a function $\mathit{order}: \bigcup_{i=0}^{n} \Sigma_i \times \NAT \to \NAT$ such that for every $0 \leq i \leq n$ if symbol $f \in \Sigma_i$, and for all $0 \leq j \leq \mathit{ar}(f)$, $\mathit{order}(f, j) = k$ if $k < i$ and we expect the $j$-th. argument of $f$ to be a term in $\mathit{Term}(\Sigma_k)$.
\end{definition}

Next we define the class of models over which we will interpret the formulae presented in Def.~\ref{hoeq} and a satisfaction relation between these models and these formulae. Note that models must be constructed over a domain able to interpret terms of any possible order; to do this, we consider the set-theoretical superstructure $S^\# = \bigcup_n V^n(S)$ where $V^1 (S) = S$ and $V^{n+1}(S) = V^n(S)\cup\mathcal{P}(V^n(S))$. For the sake of simplifying the next definition we consider the following conventions: $V^0 (S) = \emptyset$ and $V^{[\![i]\!]} = V^i / V^{i-1}$.

\begin{definition}[Interpretations and models for higher-order equational logic]
\label{def:hoeq-models}
Let ${\sf HOEq} = \<\mathsf{Sign}^{\sf HOEq}, \mathbf{Sen}^{\sf HOEq}\>$ be the language from Def.~\ref{hoeq} and $\Sigma \in |\mathsf{Sign}^{\sf HOEq}|$ be the signature $\bigcup_{i=0}^{n} \<C_i, F_i, P_i\>$; then, an \emph{interpretation} of $\Sigma$ is a structure of the form $\<S^\#, \<\emptyset, \emptyset, \overline{P_0}\>, \<\overline{C_1}, \overline{F_1}, \overline{P_1}\>, \ldots, \<\overline{C_n}, \overline{F_n}, \overline{P_n}\>\>$ such that:
\begin{itemize}[$-$]
\item for all $p \in P_0$, there is a $\overline{p} \in \overline{P_0}$ such that $\overline{p} \in \{\mathbf{true}, \mathbf{false}\}$,
\item for all $1 \leq i \leq n$: 
\begin{itemize}[$-$]
\item for all $c \in C_i$, there is a $\overline{c} \in \overline{C_i}$ such that $\overline{c} \in V^{[\![i]\!]} (S)$,
 \item for all $f \in F_i$, there is a $\overline{f} \in \overline{F_i}$ such that: $\overline{f}$ is in
 
 $\left[V^{[\![\mathit{order}(f, 1)]\!]} (S) \times \ldots \times V^{[\![\mathit{order}(f, \mathit{ar}(f))]\!]} (S) \to V^{[\![i]\!]}(S)\right]$, and 
\item for all $p \in P_i$, there is a $\overline{p} \in \overline{P_i}$ such that: $\overline{p}$ is in 

$V^{[\![\mathit{order}(p, 1)]\!]} (S) \times \ldots \times V^{[\![\mathit{order}(p, \mathit{ar}(p))]\!]} (S)$.
\end{itemize}
\end{itemize}
We will say that an interpretation is \emph{extensional} if and only if it satisfies the \emph{extensionality axiom}: \\
$$(\forall x, x')((\forall y)(y \in x \iff y \in x') \implies x = x').$$ 

Let $V = \{V_i\}_{1 \leq i \leq n}$ be sets of variable symbols, then an \emph{assignment} for $V$ is a family of mapping $v_i: V_i \to V^{[\![i]\!]}(S)$.

A structure $\<\mathcal{M}, v\>$ is a \emph{model} for $\Sigma$ and $V$ if and only if $\mathcal{M}$ is an interpretation $\Sigma$ and $\{v_i\}_{1 \leq i \leq n}$ is an assignment for $V$.

Let $V = \{V_i\}_{1 \leq i \leq n}$ be a family of sets of variable symbols then $\mathbf{Mod}^{\sf HOEq} (\Sigma)$ is the class of extensional models for $\Sigma$ and $V$. If $\sigma: \Sigma \to \Sigma' \in ||\mathsf{Sign}^{\sf HOEq}||$ then $\mathbf{Mod}^{\sf HOEq} (\sigma)$ is the natural reduct operation on $\Sigma$'-models to $\Sigma$-models.
\end{definition}

\begin{definition}[Satisfaction relation for higher-order equational logic]
\label{def:hoeq-satisfaction}
Let ${\sf HOEq} = \<\mathsf{Sign}^{\sf HOEq}, \mathbf{Sen}^{\sf HOEq}\>$ be the language from Def.~\ref{hoeq}, if $\Sigma = \bigcup_{i=0}^{n} \<C_i, F_i, P_i\> \in |\mathsf{Sign}^{\sf HOEq}|$, $V = \{V_i\}_{1 \leq i \leq n}$ be a family of sets of variable symbols and $\<\mathcal{M}, v\> \in |\mathbf{Mod}^{\sf HOEq} (\Sigma)|$ then, we define $m_{\<\mathcal{M}, v\>}: \mathit{Term}(\Sigma) \to |\mathcal{M}|$ in the following way: for all $1 \leq i \leq n$
\begin{itemize}[$-$]
\item if $t \in V_i$, then $m_{\<\mathcal{M}, v\>} (t) = v_i (t)$, 
\item if $t \in C_i$, then $m_{\<\mathcal{M}, v\>} (t) = \overline{t}$, and 
\item if $t = f(t_1, \ldots, t_{\mathit{ar}(f)})$ with $f \in F_i$, then

\hspace{0.5in}$m_{\<\mathcal{M}, v\>} (t) = \overline{f}(m_{\<\mathcal{M}, v\>} (t_1), \ldots, m_{\<\mathcal{M}, v\>} (t_{\mathit{ar}(f)}))$.
\end{itemize}

We define $\vDash^{\sf HOEq}_\Sigma \subseteq \mathbf{Mod}^{\sf HOEq} (\Sigma) \times \mathbf{Sen}^{\sf HOEq}(\Sigma)$ in the following way:
\begin{itemize}[$-$]
\item if $p \in P_0$, then $\<\mathcal{M}, v\> \vDash^{\sf HOEq}_\Sigma p$ if and only if $\overline{p}$,
\item if $t, t' \in \mathit{Term}(\Sigma)$, then $\<\mathcal{M}, v\> \vDash^{\sf HOEq}_\Sigma t = t'$ if and only if 

$m_{\<\mathcal{M}, v\>} (t) = m_{\<\mathcal{M}, v\>} (t')$,
\item for all $1 \leq i \leq n$, $p \in P_i$, and $t_1, \ldots, t_{\mathit{ar}(p)} \in \mathit{Term}(\Sigma)$, then 

$\<\mathcal{M}, v\> \vDash^{\sf HOEq}_\Sigma p (t_1, \ldots, t_{\mathit{ar}(p)})$ if and only if $\overline{p} (m_{\<\mathcal{M}, v\>} (t_1), \ldots, m_{\<\mathcal{M}, v\>} (t_{\mathit{ar}(p)}))$
\end{itemize}
\end{definition}

Finally, we define an entailment relation for formulae in the language of Def.~\ref{hoeq}.

\begin{definition}[Entailment relation for higher-order equational logic]
\label{def:hoeq-entail}
Let ${\sf HOEq} = \<\mathsf{Sign}^{\sf HOEq}, \mathbf{Sen}^{\sf HOEq}\>$ be the language from Def.~\ref{hoeq}, if $\Sigma \in |\mathsf{Sign}^{\sf HOEq}|$, $\Gamma \subseteq \mathbf{Sen}^{\sf HOEq}(\Sigma)$, $0 \leq i \leq \mathit{order}(\Sigma)$ and $t, t', t'' \in \mathit{Term}(\Sigma_i)$ and $f \in \Sigma$ and for all $1 \leq i \leq \mathit{ar(f)}$, $t_i, t'_i \in \mathit{Term}(\Sigma_{\mathit{order}(f, i)})$, then the proofs in higher-order equational logic is obtained by the application of the following rules:
$$
\begin{array}{ccc}
\mbox{
\AXC{}
\LL{[Ref]}
\UnaryInfC{$t = t$}
\DP
}
&
,
&
\mbox{
\AXC{$t = t'$}
\LL{[Sym]}
\UnaryInfC{$t' = t$}
\DP
}
\end{array}
$$
$$
\mbox{
\AXC{$t = t'$}
\AXC{$t' = t''$}
\LL{[Trans]}
\BinaryInfC{$t = t''$}
\DP
}
$$
$$
\mbox{
\AXC{$\Gamma \cup \{t = t'\}$}
\LL{[Ax]}
\UnaryInfC{$t = t'$}
\DP
}
$$
$$
\mbox{
\AXC{$\Gamma \cup \{f(t) = f'(t)\}$}
\LL{[Ext]}\RL{; for all $t \in \mathit{Term}(\Sigma)$}
\UnaryInfC{$f = f'$}
\DP
}
$$
$$
\mbox{
\AXC{$t_1 = t'_1$}
\AXC{$\cdots$}
\AXC{$t_k = t'_k$}
\LL{[Func]}
\TrinaryInfC{$f(t_1, \ldots, t_k) = f(t_1, \ldots, t_k)$}
\DP
}
$$
Then, if $\alpha \in \mathbf{Sen}^{\sf HOEq}(\Sigma)$, a proof of $\alpha$ from the set of hypotheses $\Gamma$ is a tree-like structure formed by the application of the previous rules and is denoted $\Gamma \vdash^{\sf HOEq}_\Sigma \alpha$.
\end{definition}

\begin{theorem}[Higher-order equational logic]
Let ${\sf HOEq} = \<\mathsf{Sign}^{\sf HOEq}, \mathbf{Sen}^{\sf HOEq}\>$ be the language from Def.~\ref{hoeq}, $\mathbf{Mod}^{\sf HOEq}: {\mathsf{Sign}^{\sf HOEq}}^\op \to \mathsf{Cat}$ be the functor from Def.~\ref{def:hoeq-models}, $\{\vDash^{\sf HOEq}_{\Sigma}\}_{\Sigma \in |\mathsf{Sign}^{\sf HOEq}|}$ the family of consequence relations from Def.~\ref{def:hoeq-satisfaction}, and $\{\vdash^{\sf HOEq}_{\Sigma}\}_{\Sigma \in |\mathsf{Sign}^{\sf HOEq}|}$ the family of entailment relations from Def.~\ref{def:hoeq-entail}, then:
\begin{enumerate}
\item \label{thm:institution} $\< \mathsf{Sign}^{\sf HOEq}, \mathbf{Sen}^{\sf HOEq}, \mathbf{Mod}^{\sf HOEq}, \{\models^{\sf HOEq}_{\Sigma}\}_{\Sigma \in |\mathsf{Sign}^{\sf HOEq}|}\>$ is an institution,
\item \label{thm:entailment} $\< \mathsf{Sign}^{\sf HOEq}, \mathbf{Sen}^{\sf HOEq}, \{\vdash^{\sf HOEq}_{\Sigma}\}_{\Sigma \in |\mathsf{Sign}^{\sf HOEq}|}\>$ is an entailment system, and
\item \label{thm:sound+complete} $\< \mathsf{Sign}^{\sf HOEq}, \mathbf{Sen}^{\sf HOEq}, \mathbf{Mod}^{\sf HOEq}, \{\vdash^{\sf HOEq}_{\Sigma}\}_{\Sigma \in |\mathsf{Sign}^{\sf HOEq}|}, \right.$\\
$\left. \{\models^{\sf HOEq}_{\Sigma}\}_{\Sigma \in |\mathsf{Sign}^{\sf HOEq}|}\>$ is a sound and complete logic.
\end{enumerate}
\end{theorem}  
\begin{proof}
The proofs of parts~\ref{thm:institution}~and~\ref{thm:entailment} follow directly from Defs.~\ref{hoeq},~\ref{def:hoeq-models},~\ref{def:hoeq-satisfaction}~and~\ref{def:hoeq-entail} and are analogous to the many examples of institution and entailment system definitions in the literature. The proof for part~\ref{thm:sound+complete} follows from \cite[Thm.~2.8]{meinke:tcs-100_2}.
\end{proof}

This concludes the presentation of higher order equational logic as the language in which the notion of state (or world) is formalised so it can be used in the definition of satisfiability of the formulae in the state sublanguage. Still, the same observations apply to the characterisation of the transition systems over the states. We chose a specific relational language with reflexive and transitive closure as it is a perfect fit for the problem of formally characterising binary relations but it is, by no means, the only option available. Also note that the approach works independently of the particular language of choice, as far as the language can be proved to be a complete logic, in the sense of Def.~\ref{logic}.

\subsection{Discussion}
\label{subsec:discussion}
It is natural to raise a discussion regarding the reaches and limitations of the framework presented above. An important issue, revealed by the previous section, is the problem of the appropriateness of a class of algebras as a representation of the intended relational models for a given logical language. The definitions above provide a general and concrete way of formalising the relational models of logical languages by interpreting them over classes of algebras of relations. Still, even when the abstract model theory of a language cannot serve as the means of analysis or refinement of logical descriptions, the question on how concrete and abstract model theory relate to each other is of relatively big importance. Regarding this, one must observe that different classes of relational models correspond to different restrictions in the relations intervening in the definition of the class of algebras; only those who are axiomatisable in {\sf ETR*} can aspire to have a concrete counterpart defined within the framework presented in this paper. For example, in the hypothetical case we had used {\sf ETR} \cite{tarski:jsl-6_3} as a formal language for describing the relational structure of models, then linear temporal logics could be given semantics while dynamic logics could not (see Ex.~\ref{ex:satisfaction}) as the reflexive and transitive closure cannot be axiomatised in a first-order logic, which is the logical structure underlying {\sf ETR}. This limitation of the framework also reveals a positive element pointing at its generality; the selection of languages for describing states and the relational structure of models can be done in a modular way depending on the needs, as far as it is possible to prove a result analogous to Thm.~\ref{thm:complete}.

An aspect of upmost importance to consider is that many logics, among which we find some deontic logics \cite{castro:ictac07} and substructural logics \cite{paoli02}, are given semantics in such a way that the interpretation of modal operators changes depending on the state in which the formula is being evaluated. A logic with this type of feature requires extending the framework by:
\begin{inparaenum}[1.]
\item definiing a relational signature in which it is possible to distinguish rigid symbols and flexible symbols, analogously to the role played by functors $\mathbf{RSign}$ and $\mathbf{FSign}$, and
\item defining a notion of relational state analogous to the role played by functor $\mathbf{Sts}$.
\end{inparaenum}
These aspects will remain as a further line of research, but there are doubts raised by the fact that such a change would require a relational language together with a corresponding class of models with infinite formulae of higher cardinality putting the existence of a sound and complete calculus at risk.

.
\section{Examples}
\label{examples}
In this section we present some additional examples on how this framework provide, concrete classes of models for different modal logics. 

\subsection{Linear temporal logic}
\label{ex:ltl}
Let $\<\mathsf{Sign}^{\sf LTL}, \mathbf{Sen}^{\sf LTL}\>$ be the language of linear temporal logic defined as in \cite{pnueli:tcs-13_1}. It's syntax and some brief considerations about it's semantics are summarised below.

Let $\Sigma = \{p_i\}_{i \in \mathcal{I}}$ a set of propositional flexible symbols, then $\mathit{FormLTL} (\Sigma)$ is the smallest set such that:
\begin{itemize}[$-$] 
\item $\Sigma \subseteq \mathit{FormLTL} (\Sigma)$, and
\item if $\alpha, \beta \in \mathit{FormLTL} (\Sigma)$, then $\{\neg\alpha, \alpha\lor\beta, {\sf X}\ \alpha, \alpha\ {\sf U}\ \beta\} \subseteq \mathit{FormLTL} (\Sigma)$.
\end{itemize}

The rest of the boolean operators are defined as usual (i.e. $\alpha \land \beta = \neg\(\neg\alpha\lor\neg\beta\)$, $\alpha \implies \beta = \neg\alpha \lor \beta$, etc.) and the rest of the ${\sf LTL}$ temporal operators are defined as follows:
\begin{itemize}[$-$]
\item ${\sf F}\ \phi = \text{true}\ {\sf U}\ \phi$: eventually $\phi$ is true,
\item ${\sf G}\ \phi = \neg\( {\sf F}\ \(\neg\phi\)\)$: $\phi$ always remain true,
\item $\phi\ {\sf R}\ \psi = \neg\( \neg\phi\ {\sf U}\ \neg\psi \)$: known as \emph{release}, $\psi$ remains true until and including once $\phi$ becomes true,
\item $\phi\ {\sf W}\ \psi = \(\phi\ {\sf U} \psi\) \lor {\sf G}\ \phi$: known as \emph{weak until}, $\phi$ remains true, either until $\phi$ becomes true, or forever, and
\item $\phi\ {\sf M}\ \psi = \(\phi\ {\sf R} \psi\) \lor {\sf F}\ \phi$: known as \emph{strong release}.
\end{itemize}

The semantics of ${\sf LTL}$ formulae over a signature $\{p_i\}_{i \in \mathcal{I}}$ is given in terms of infinite paths in a Kripke structure $\mathfrak{K} = \<W, W_0, T, \mathcal{L}\>$ such that:
\begin{inparaenum}[a)]
\item $W_0 \subseteq W$
\item $T \subseteq W \times W$, and
\item $\mathcal{L}: \{p_i\}_{i \in \mathcal{I}} \to 2^W$.
\end{inparaenum}
If $\pi = w_0, w_1, \ldots$ is an infinite sequence such that:
\begin{inparaenum}[a)]
\item for all $i \in \NAT$, $w_i \in W$,
\item $w_0 \in W_0$, and
\item for all $i \in \NAT$, $w_i\ T\ w_{i+1}$;
\end{inparaenum}
and $\alpha \in \mathit{FormLTL} (\{p_i\}_{i \in \mathcal{I}})$, then $\alpha$ is satisfied by $\pi$ in $\mathfrak{K}$ (written $\mathfrak{K}, \pi \models \alpha$) is inductively defined as follows:
\[
\begin{array}{l}
\mathfrak{K}, \pi \models p_i \mbox{ iff } \pi[0] \in \mathcal{L} (p_i)\\
\mathfrak{K}, \pi \models \neg\alpha \mbox{ iff } \mathfrak{K}, \pi \models \alpha \mbox{ does not hold}\\
\mathfrak{K}, \pi \models \alpha\lor\beta \mbox{ iff } \mathfrak{K}, \pi \models \alpha \text{ or } \mathfrak{K}, \pi \models \beta \\
\mathfrak{K}, \pi \models {\sf X}\ \alpha \mbox{ iff } \mathfrak{K}, \pi[1..] \models \alpha \\
\mathfrak{K}, \pi \models \alpha\ {\sf U}\ \beta \mbox{ iff } \mbox{there exists $j \in \NAT$ such that: }\\
\hspace{0.7in} \mbox{$\mathfrak{K}, \pi[j..] \models \beta$ and for all $i \leq j$,  $\mathfrak{K}, \pi[i..] \models \alpha$}
\end{array}
\]
\noindent where $\pi[i]$ denotes the $i^{th}$ element of the sequence $\pi$ and $\pi[i..]$ denotes the subsequence of $\pi$ starting in the $i^{th}$ position.

The state sublanguage $\<\mathbf{StSen}, \mathbf{RSign}, \mathbf{FSign}, \rho^{Sen}\>$ for ${\sf LTL}$ is defined as: let $\{p_i\}_{i \in \mathcal{I}} \in |\mathsf{Sign}^{\sf LTL}|$ and $\sigma: \{p_i\}_{i \in \mathcal{I}} \to \{p'_i\}_{i \in \mathcal{I}'} \in ||\mathsf{Sign}^{\sf LTL}||$ such that $\sigma (p_i) = (p'_i)$,
\begin{itemize}[$-$]
\item $\mathbf{StSen} (\{p_i\}_{i \in \mathcal{I}}) = \{p_i\}_{i \in \mathcal{I}}$, $\mathbf{StSen} (\sigma) = \{\<p_i, p'_i\>\}_{i \in \mathcal{I}}$,
\item $\mathbf{RSign} (\{p_i\}_{i \in \mathcal{I}}) = \<\emptyset, \emptyset, \emptyset\>$, $\mathbf{StSen} (\sigma) = \mathit{id}_{\<\emptyset, \emptyset, \emptyset\>}$,
\item $\mathbf{FSign} (\{p_i\}_{i \in \mathcal{I}}) = \<\emptyset, \emptyset, \{P_i\}_{i \in \mathcal{I}}\>$, 

\hfill $\mathbf{StSen} (\sigma) = \{\<P_i, P'_i\>\}_{i \in \mathcal{I}'}$, and
\item $\rho^{Sen}_{\{p_i\}_{i \in \mathcal{I}}} (p_i) = P_i$, for all $i \in \mathcal{I}$.
\end{itemize}

The functors $\mathbf{Ints}$ and $\mathbf{Sts}$ are defined as in Defs.~\ref{def:interpretations}~and~\ref{def:states}, Thus, providing a characterisation of the set of states $\mathbf{States}^{\sf LTL}$.

Let $\Sigma \in |\mathsf{Sign}^{\sf LTL}|$; then, we define $\mathbf{RelTh}(\Sigma) = \<\Sigma^{\mathit{Rel}}, \Gamma^{\mathit{Rel}}\>$ as follows:
$$
\begin{array}{rcl}
\Sigma^{\mathit{Rel}} & = & \<St_0, T\> \\
\Gamma^{\mathit{Rel}} & = & \{\(\forall x, y\)\(x\ St_0\ y \implies x\ \aid\ y\) \\
                                     &     & \qquad\qquad \mbox{[There is a set of initial states]}\\
                                     &     & \ \ \(\forall x, y\)\(x\ \(T \acompo \aunit\) \ameet \aid\ y \iff x\ \aid\ y\) \\
                                     &     & \qquad\qquad \mbox{[$T$ is total]}\\
                                     &     & \ \ \(\forall x, y\)\(x\ \aconv{T}\acompo T\ y \implies x\ \aid\ y\) \\
                                     &     & \qquad\qquad \mbox{[$T$ is functional]} \qquad\qquad\qquad \}\\
\end{array}
$$

Given $\{p_i\}_{i \in \mathcal{I}} \in |\mathsf{Sign}^{\sf LTL}|$, $\mathbf{Mod}^{\sf LTL} (\{p_i\}_{i \in \mathcal{I}}) = \<\mathcal{O}, \mathcal{A}\>$  where\footnote{As we mentioned in previous sections, and for the sake of the compactness of the presentation, we resort to the formalisation of the generic datatype ${\it List}$ (see \cite[pp.~55--56]{backhouse:afp98}).}:
\begin{itemize}[$-$]
\item $\mathcal{O} = \left\{\<I, \mathcal{M}, \pi\>\ {\Big |}\ I \in |\mathbf{Ints} (\Sigma)|, \mathcal{M} \in |\mathbf{RelStr} (\Sigma)| \text{ and } \right.$

\hfill $\left. \pi \in |{\it List}\ \(\mathbf{States} (\Sigma)|_{\mathcal{M}_{bs}}\)| \text{ such that } \mbox{ for all } i \in \NAT, \pi[i]\ T^\mathcal{M}\ \pi[i+1] \right\} $
\item $\mathcal{A} = \left\{\<\sigma, \gamma_h, h\>: \<I, \mathcal{M}, \pi\> \to \<I', \mathcal{M}', \pi'\>\ {\Big |}\ \sigma: I \to I' \in ||\mathbf{Ints} (\Sigma)||, \right.$

\hfill $\left. \gamma_h:\mathcal{M} \to \mathcal{M}' \in ||\mathbf{RelStr} (\Sigma)||, \text{ and } h: \pi \to \pi' \in ||{\it List}\ \(\mathbf{States} (\Sigma)|_{\mathcal{M}_{bs}}\)|| \right\} $
\end{itemize}

Let $\Sigma \in |\mathsf{Sign}^{\sf LTL}|$, we extend $\models^{\sf L}_\Sigma$ from Def.~\ref{def:satisfiability} as follows:
\[
\begin{array}{l}
\<I, \mathcal{M}, \pi\> \models^{\sf LTL}_\Sigma \alpha \mbox{ iff } \<I, \pi[0]\> \models^{\sf LTL}_\Sigma \alpha \\
\hfill \mbox{, for all $\alpha \in \mathbf{StSen} (\Sigma)$}\\
\<I, \mathcal{M}, \pi\> \models^{\sf LTL}_\Sigma \neg\alpha \mbox{ iff } \<I, \mathcal{M}, \pi\> \models^{\sf LTL}_\Sigma \alpha \mbox{ does not hold}\\
\<I, \mathcal{M}, \pi\> \models^{\sf LTL}_\Sigma \alpha \lor \beta \mbox{ iff } \<I, \mathcal{M}, \pi\> \models^{\sf LTL}_\Sigma \alpha \mbox{ or } \<I, \mathcal{M}, \pi\> \models^{\sf LTL}_\Sigma \beta\\
\<I, \mathcal{M}, \pi\> \models^{\sf LTL}_\Sigma \X\alpha \mbox{ iff } \<I, \mathcal{M}, \pi[1..]\>  \models^{\sf LTL}_\Sigma \alpha\\
\<I, \mathcal{M}, \pi\> \models^{\sf LTL}_\Sigma \alpha\U\beta \mbox{ iff } \mbox{there exists $i \in \NAT$ such that:}\\
\hfill \begin{array}{l} \<I, \mathcal{M}, \pi[i..]\> \models^{\sf LTL}_\Sigma \beta \mbox{, and} \\
                                \mbox{for all $j \in \NAT$, $j < i$ implies $\<I, \mathcal{M}, \pi[j..]\> \models^{\sf LTL}_\Sigma \alpha$.} \end{array}
\end{array}
\]

The reader should note that both $\pi[i]$ and $\pi[i..]$ must be defined within the formal framework chosen for defining ${\sf Struct}^{\sf LTL}$. In the case of our presentation $\bullet [\bullet]:{\sf Struct}^{\sf LTL}_\Sigma \to \mathbf{States} (\Sigma)$ is defined as a functor and $\bullet [\bullet..]:{\sf Struct}^{\sf LTL}_\Sigma \to {\sf Struct}^{\sf LTL}_\Sigma$ as an endofunctor according to the behaviour described above.

\subsection{Computational tree logic}
\label{ex:ctl}
Let $\<\mathsf{Sign}^{\sf CTL}, \mathbf{Sen}^{\sf CTL}\>$ be the language of computational tree logic defined as in \cite{benari:acm-sigplan-sigact81}. As we did in the previous example, we summarise it's syntax and semantics below.

Let $\Sigma = \{p_i\}_{i \in \mathcal{I}}$ a set of propositional flexible symbols, then $\mathit{FormCTL} (\Sigma)$ is the smallest set such that:
\begin{itemize}[$-$] 
\item $\Sigma \subseteq \mathit{FormCTL} (\Sigma)$, and
\item if $\alpha, \beta \in \mathit{FormCTL} (\Sigma)$, then $\{\neg\alpha, \alpha\lor\beta, {\sf EX}\ \alpha, {\sf EG}\ \alpha, $\\
${\sf E}[\alpha\ {\sf U}\ \beta]\} \subseteq \mathit{FormCTL} (\Sigma)$.
\end{itemize}

The rest of the boolean operators are defined as usual and the rest of the ${\sf CTL}$ temporal operators are defined as follows:
\begin{itemize}[$-$]
\item ${\sf EF}\ \phi = {\sf E}[\text{true}\ {\sf U}\ \phi]$: there exists an execution from the current state in which eventually $\phi$ is true,
\item ${\sf EG}\ \phi = {\sf E}[\neg \(\text{false}\ {\sf U}\ \neg\phi\)]$: there exists an execution from the current state in which $\phi$ is always true,
\item ${\sf AX}\ \phi = \neg {\sf EX}(\neg\phi)$: in every successor of the current state $\phi$ is true,
\item ${\sf A}[\phi\ {\sf U}\ \psi] = \neg\({\sf E}[\(\neg\psi\)\ {\sf U}\ \neg\(\phi\lor\psi\)] \lor {\sf EG}\ \(\neg\psi\)\)$: in every execution from the current state $\phi$ is true until $\psi$ becomes true,
\item ${\sf AF}\ \phi = \neg {\sf EG}(\neg \phi)$: in every execution from the current state eventually $\phi$ is true, and
\item ${\sf AG}\ \phi = \neg\({\sf EF}\ \(\neg\phi\)\)$: in every execution from the current state $\phi$ is always true.
\end{itemize}

The semantics of ${\sf CTL}$ formulae over a signature $\{p_i\}_{i \in \mathcal{I}}$ is given in terms of infinite paths in a Kripke structure $\mathfrak{K} = \<W, W_0, T, \mathcal{L}\>$ such that:
\begin{inparaenum}[a)]
\item $W_0 \subseteq W$
\item $T \subseteq W \times W$, and
\item $\mathcal{L}: \{p_i\}_{i \in \mathcal{I}} \to 2^W$.
\end{inparaenum}
If $\alpha \in \mathit{FormCTL} (\{p_i\}_{i \in \mathcal{I}})$, then $\alpha$ is satisfied by $w \in W$ in $\mathfrak{K}$ (written $\mathfrak{K}, w \models \alpha$) is inductively defined as follows:
\[
\begin{array}{l}
\mathfrak{K}, w \models p_i \mbox{ iff } w \in \mathcal{L} (p_i)\\
\mathfrak{K}, w \models \neg\alpha \mbox{ iff } \mathfrak{K}, w \models \alpha \mbox{ does not hold}\\
\mathfrak{K}, w \models \alpha\lor\beta \mbox{ iff } \mathfrak{K}, w \models \alpha \mbox{ or } \mathfrak{K}, w \models \beta\\
\mathfrak{K}, w \models {\sf EX}\ \alpha \mbox{ iff } \mbox{there exists $w' \in W$ such that:} \mbox{$w\ T\ w'$ and $\mathfrak{K}, w' \models \alpha$}\\
\mathfrak{K}, w \models {\sf EG}\ \alpha \mbox{ iff } \mbox{there exists $\pi \in W^\omega$ such that:}\\
\hfill \pi[0] = w \mbox{ and for all $i \in \NAT$, $\pi[i]\ T\ \pi[i+1]$, and} \\
\hfill \mbox{for all $i \in \NAT$, $\mathfrak{K}, \pi[i] \models \alpha$}\\
\mathfrak{K}, w \models {\sf E}[\alpha\ {\sf U}\ \beta] \mbox{ iff } \mbox{there exists $\pi \in W^\omega$ such that:}\\
\hfill \begin{array}{l}
\pi[0] = w \mbox{ and for all $i \in \NAT$, $\pi[i]\ T\ \pi[i+1]$, and}\\
 \mbox{there exists $i \in \NAT$ such that:}\\
 \mbox{$\mathfrak{K}, \pi[i] \models \beta$ and for all $j \in \NAT$, $j < i$ implies $\mathfrak{K}, \pi[j] \models \alpha$} \end{array}
\end{array}
\]
\noindent where $W^\omega$ denotes the set of infinite sequences of elements from $W$, $\pi[i]$ denotes the $i^{th}$ element of the sequence $\pi$.

The state sublanguage $\<\mathbf{StSen}, \mathbf{RSign}, \mathbf{FSign}, \rho^{Sen}\>$ for ${\sf CTL}$, and it's relational theory $\mathbf{RelTh}(\Sigma) = \<\Sigma^{\mathit{Rel}}, \Gamma^{\mathit{Rel}}\>$ are defined as in Ex.~\ref{ex:ltl}.

Given $\Sigma \in |\mathsf{Sign}^{\sf CTL}|$, $\mathbf{Mod}^{\sf CTL} (\Sigma) = \<\mathcal{O}, \mathcal{A}\>$ where:
\begin{itemize}[$-$]
\item $\mathcal{O} = \left\{ \<I, \mathcal{M}, s\> \ {\Big |}\ I \in |\mathbf{Ints} (\Sigma)|, \mathcal{M} \in |\mathbf{RelStr} (\Sigma)|, s \in |\mathbf{States} (\Sigma)|_{\mathcal{M}_{bs}}| \right\}$
\item $\mathcal{A} = \left\{ \<\sigma, \gamma_h, h\>: \<I, \mathcal{M}, s\> \to \<I', \mathcal{M}', s'\>\ {\Big |}\ \sigma: I \to I' \in ||\mathbf{Ints} (\Sigma)||,\right.$

\hfill $\left. \gamma_h:\mathcal{M} \to \mathcal{M}' \in ||\mathbf{RelStr} (\Sigma)||, \text{ and } h: s \to s' \in ||\mathbf{States} (\Sigma)|_{\mathcal{M}_{bs}}|| \right\}$
\end{itemize}

Let $\Sigma \in |\mathsf{Sign}^{\sf CTL}|$, we extend $\models^{\sf L}_\Sigma$ from Def.~\ref{def:satisfiability} as follows:
\[
\begin{array}{l}
\<I, \mathcal{M}, s\> \models^{\sf CTL}_\Sigma \alpha \mbox{ iff } \<I, s\> \models^{\sf CTL}_\Sigma \alpha \mbox{, for all $\alpha \in \mathbf{StSen} (\Sigma)$}\\
\<I, \mathcal{M}, s\> \models^{\sf CTL}_\Sigma \neg\alpha \mbox{ iff } \<I, \mathcal{M}, \tau\> \models^{\sf CTL}_\Sigma \alpha \mbox{ does not hold}\\
\<I, \mathcal{M}, s\> \models^{\sf CTL}_\Sigma \alpha \lor \beta \mbox{ iff } \<I, \mathcal{M}, \tau\> \models^{\sf CTL}_\Sigma \alpha \mbox{ or } \<I, \mathcal{M}, \tau\> \models^{\sf CTL}_\Sigma \beta\\
\<I, \mathcal{M}, s\> \models^{\sf CTL}_\Sigma \Ex\alpha \mbox{ iff }
\mbox{there exists $s' \in |\mathbf{States} (\Sigma)|_{\mathcal{M}_{bs}}|$ such that:}\\
\hfill \mbox{$s\ T^\mathcal{M}\ s'$ and $\<I, \mathcal{M}, s'\> \models^{\sf CTL}_\Sigma \alpha$}\\
\<I, \mathcal{M}, s\> \models^{\sf CTL}_\Sigma \EG\alpha \mbox{ iff }
\mbox{there exists $\pi \in |({\it List}\ \(\mathbf{States} (\Sigma)|_{\mathcal{M}_{bs}}\)|$, such that:}\\
\hfill \begin{array}{l}
\pi[0] = s \mbox{ and for all $i \in \NAT$, $\pi[i]\ T^\mathcal{M}\ \pi[i+1]$, and }\\
\mbox{for all $i \in \NAT$, $\<I, \mathcal{M}, \pi[i]\> \models^{\sf CTL}_\Sigma \alpha$.} \end{array}\\
\<I, \mathcal{M}, s\> \models^{\sf CTL}_\Sigma \eee[\alpha\U\beta] \mbox{ iff }
\mbox{there exists $\pi \in |({\it List}\ \(\mathbf{States} (\Sigma)|_{\mathcal{M}_{bs}}\)|$, such that:}\\
\hfill \begin{array}{l}
\pi[0] = s \mbox{ and for all $i \in \NAT$, $\pi[i]\ T^\mathcal{M}\ \pi[i+1]$, and}\\
\mbox{there exists $i \in \NAT$ such that $\<I, \mathcal{M}, \pi[i]\> \models^{\sf CTL}_\Sigma \beta$ and}\\
\mbox{for all $j \in \NAT$, $j < i$ implies $\<I, \mathcal{M}, \pi[j]\> \models^{\sf CTL}_\Sigma \alpha$.}\end{array}\\
\end{array}
\]

Once again, the reader should note that both $\pi[i]$ and $\pi[i..]$ must be defined within the formal framework chosen for defining ${\sf Struct}^{\sf CTL}$ as we did in Ex.~\ref{ex:ltl}.

\subsection{First order computational tree logic $*$}
Let $\<\mathsf{Sign}^{\sf FOCTL*}, \mathbf{Sen}^{\sf FOCTL*}\>$ be the language of a first order version of computational tree logic $*$ (see \cite{emerson:jcss-30_1,emerson:jacm-33_1} for the definition of computational tree logic $*$). As we did in the previous examples, we summarise it's syntax and semantics below.

Let $\Sigma = \<C, F, P, A\> \in |{\sf Sign}^{\sf FOCTL*}|$ and $X$ a set of first order variable symbols, then $\mathit{TermFODL} (\Sigma)$ is the smallest set such that:
\begin{itemize}[$-$]
\item $C \cup X \subseteq \mathit{TermFOCTL*} (\Sigma)$, and
\item if $f \in F$ and $\{t_1, \ldots, t_{\mathit{arity} (f)}\} \subseteq \mathit{TermFOCTL*} (\Sigma)$, then $f (t_1, \ldots, t_2) \in \mathit{TermFOCTL*} (\Sigma)$
\end{itemize}
Next, we mutually define state and path formulae, as the smallest sets $\mathit{FormStFOCTL*} (\Sigma)$ and $\mathit{FormPtFOCTL*} (\Sigma)$ such that:
\begin{itemize}[$-$]
\item if $t_1, t_2 \in \mathit{TermFOCTL*} (\Sigma)$, then $t_1 = t_2 \in \mathit{FormStFOCTL*} (\Sigma)$, 
\item if $p \in P$ and $\{t_1, \ldots, t_{\mathit{arity} (p)}\} \subseteq \mathit{TermFOCTL*} (\Sigma)$, then $p (t_1, \ldots, t_2) \in \mathit{FormStFOCTL*} (\Sigma)$,
\item if $x \in X$, $\varphi \in \mathit{FormPtFOCTL*} (\Sigma)$ and $\alpha, \beta \in \mathit{FormStFOCTL*} (\Sigma)$, then $\{\neg\alpha, \alpha\lor\beta, (\exists x)\alpha, {\sf E}\ \varphi\} \subseteq \mathit{FormStFOCTL*} (\Sigma)$, and
\item if $\alpha \in \mathit{FormStFOCTL*} (\Sigma)$ and $\varphi, \psi \in \mathit{FormPtFOCTL*} (\Sigma)$, then $\{\alpha, \neg\varphi, \varphi\lor\psi, {\sf X}\ \varphi, \varphi\ {\sf U}\ \psi\} \subseteq \mathit{FormPtFOCTL*} (\Sigma)$, 
\end{itemize}

The additional logical operators are defined as usual, the additional temporal operators yielding state formulae are defined as in Ex.~\ref{ex:ctl}, and those yielding path formulae are defined as in Ex.~\ref{ex:ltl}.

The semantics of ${\sf FOCTL*}$ formulae over a signature $\<C, F, P\>$ and a set of first order variable symbols $X$ is given, as usual, in terms of an interpretation $\<S, \{\overline{c}\}_{c \in C}, \{\overline{f}\}_{f \in F}, \{\overline{p}\}_{p \in P}\>$ of the first order signature $\<C, F, P\>$, satisfying:
\begin{inparaenum}[1)]
\item $\overline{c} \in S$,
\item $\overline{f} \in [S^{\mathit{arity} (f)} \to S]$, for all $f \in F$, and
\item $\overline{p} \subseteq S^{\mathit{arity} (p)} \times S$, for all $p \in P$,
\end{inparaenum}
\noindent and states of a Kripke structure $\mathfrak{K} = \<W, W_0, \mathcal{L}\>$ such that:
\begin{inparaenum}[a)]
\item $W_0 \subseteq W$, and
\item $\mathcal{L}: W \to [X \to S]$.
\end{inparaenum}
Therefore, if $\alpha \in \mathit{FormStFOCTL*} (\<C, F, P\>)$, then $\alpha$ is satisfied by $w \in W$ in $\mathfrak{K}$ (written $\mathfrak{K}, w \models \alpha$) is inductively defined as follows:
\[
\begin{array}{l}
\mathfrak{K}, w \models t_1 = t_2 \mbox{ iff } m_{\<\mathfrak{K}, w\>} (t_1) = m_{\<\mathfrak{K}, w\>} (t_2)\\
\mathfrak{K}, w \models p (t_1, \ldots, t_{\mathit{arity} (p)}) \mbox{ iff } \<m_{\<\mathfrak{K}, w\>} (t_1), \ldots, m_{\<\mathfrak{K}, w\>} (t_{\mathit{arity} (p)})\> \in \overline{P}\\
\mathfrak{K}, w \models \neg\alpha \mbox{ iff } \mathfrak{K}, w \models \alpha \mbox{ does not hold}\\
\mathfrak{K}, w \models \alpha\lor\beta \mbox{ iff } \mathfrak{K}, w \models \alpha \mbox{ or } \mathfrak{K}, w \models \beta \\
\mathfrak{K}, w \models (\exists x) \alpha \mbox{ iff } \mbox{there exists $w' \in W$, $s \in S$ such that:}\\
\hfill \begin{array}{l}
\mathcal{L} (w') = \mathcal{L} (w)[x \mapsto s] \mbox{, and}\\
\mathfrak{K}, w' \models \alpha \end{array}\\
\mathfrak{K}, w \models {\sf E}\ \varphi \mbox{ iff } \mbox{there exists $\pi \in W^\omega$ such that $\mathfrak{K}, \pi \models \varphi$}\\
\ \\
\mathfrak{K}, \pi \models \alpha \mbox{ iff } \mathfrak{K}, \pi[0] \models \alpha \mbox{, $\alpha \in \mathit{FormStFOCTL*} (\<C, F, P\>)$}\\
\mathfrak{K}, \pi \models \neg\varphi \mbox{ iff } \mathfrak{K}, \pi \models \varphi \mbox{ does not hold}\\
\mathfrak{K}, \pi \models \varphi\lor\psi \mbox{ iff } \mathfrak{K}, \pi \models \varphi \mbox{ or } \mathfrak{K}, \pi \models \psi \\
\mathfrak{K}, \pi \models {\sf X}\ \varphi \mbox{ iff } \mathfrak{K}, \pi[1..] \models \varphi \\
\mathfrak{K}, \pi \models \varphi\ {\sf U}\ \psi \mbox{ iff } \mbox{there exists $j \in \NAT$ such that:}\\
\hfill \mbox{$\mathfrak{K}, \pi[j..] \models \psi$ and for all $i \leq j$, $\mathfrak{K}, \pi[i..] \models \varphi$}
\end{array}
\]
\noindent where $W^\omega$ denotes the set of infinite sequences of elements from $W$, $\pi[i]$ denotes the $i^{th}$ element of the sequence $\pi$.
\[
\begin{array}{l}
m_{\<\mathfrak{K}, w\>} (c) = \overline{c}\\
m_{\<\mathfrak{K}, w\>} (x) = \mathcal{L} (w) (x)\\
m_{\<\mathfrak{K}, w\>} (f (t_1, \ldots, t_{\mathit{arity}(f)})) =  \overline{f} (m_{\<\mathfrak{K}, w\>} (t_1), \ldots, m_{\<\mathfrak{K}, w\>} (t_{\mathit{arity}(f)}))
\end{array}
\]
The state sublanguage $\<\mathbf{StSen}, \mathbf{RSign}, \mathbf{FSign}, \rho^{Sen}\>$ for ${\sf FOCTL*}$ is defined exactly as in Ex.~\ref{ex:state-sublanguage}. On the other hand, $\mathbf{RelTh}(\Sigma) = \<\Sigma^{\mathit{Rel}}, \Gamma^{\mathit{Rel}}\>$ is defined as in Ex.~\ref{ex:ltl}.

Given $\Sigma \in |\mathsf{Sign}^{\sf FOCTL*}|$, $\mathbf{Mod}^{\sf FOCTL*} (\Sigma) = \<\mathcal{O}, \mathcal{A}\>$  where\footnote{As we did in the previous example, we resort to the formalisation of the generic datatype ${\it List}$ (see \cite[pp.~55--56]{backhouse:afp98}).}:
\begin{itemize}[$-$]
\item $\mathcal{O} = \left\{\<I, \mathcal{M}, x\>\ {\Big |}\ I \in |\mathbf{Ints} (\Sigma)|, \mathcal{M} \in |\mathbf{RelStr} (\Sigma)| \text{ and } \right.$

\hfill $x \in |\mathbf{States} (\Sigma) + {\it List}\ \(\mathbf{States} (\Sigma)|_{\mathcal{M}_{bs}}\)| \text{ such that}$

\hfill $\left. \text{if } x = in_r (\pi) \text{ then for all } i \in \NAT, \pi[i]\ T^\mathcal{M}\ \pi[i+1] \right\}$,
\item $\mathcal{A} = \left\{\<\sigma, \gamma_h, h\>: \<I, \mathcal{M}, x\> \to \<I', \mathcal{M}', x'\>\ |\ \sigma: I \to I' \in ||\mathbf{Ints} (\Sigma)||, \right.$

\hfill $\begin{array}{l}
\gamma_h:\mathcal{M} \to \mathcal{M}' \in ||\mathbf{RelStr} (\Sigma)||, \\
\left. h: x \to x' \in ||\mathbf{States} (\Sigma)|_{\mathcal{M}_{bs}} + {\it List}\(\mathbf{States} (\Sigma)|_{\mathcal{M}_{bs}}\)|| \right\} \end{array}$
\end{itemize}

Let $\Sigma \in |\mathsf{Sign}^{\sf FOCTL*}|$, we extend $\models^{\sf L}_\Sigma$ from Def.~\ref{def:satisfiability} as follows:
\[
\begin{array}{l}
\<I, \mathcal{M}, in_l (s)\> \models^{\sf FOCTL*}_\Sigma \alpha \mbox{ iff } \<I, s\> \models^{\sf FOCTL*}_\Sigma \alpha \mbox{, for all $\alpha \in \mathbf{StSen} (\Sigma)$}\\
\<I, \mathcal{M}, in_l (s)\> \models^{\sf FOCTL*}_\Sigma \neg\alpha \mbox{ iff }  \<I, \mathcal{M}, in_l (s)\> \models^{\sf FOCTL*}_\Sigma \alpha \mbox{ does not hold}\\
\<I, \mathcal{M}, in_l (s)\> \models^{\sf FOCTL*}_\Sigma \alpha \lor \beta \mbox{ iff } \\
\hfill \<I, \mathcal{M}, in_l (s)\> \models^{\sf FOCTL*}_\Sigma \alpha \mbox{ or } \<I, \mathcal{M}, in_l (s)\> \models^{\sf FOCTL*}_\Sigma \beta\\
\<I, \mathcal{M}, in_l (s)\> \models^{\sf FOCTL*}_\Sigma (\exists x)\alpha \mbox{ iff } \mbox{there exists $s' \in |\mathbf{States} (\Sigma)|_{\mathcal{M}_{bs}}|$ such that if:}\\
- \mbox{the pair $\<\mathit{in}_{l}: I \to T, \mathit{id}_{(\mathbf{RSign}+\mathbf{FSign})(\Sigma)}: s \to T\>$ is}\\
\mbox{the pushout for $\<id_{\mathbf{RSign}(\Sigma)}:\<\mathbf{RSign}(\Sigma), \emptyset\> \to I, \mathit{in}_{l}: \<\mathbf{RSign}(\Sigma), \emptyset\> \to s\>$, and}\\
- \mbox{the pair $\<\mathit{in}_{l}: I \to T', \mathit{id}_{(\mathbf{RSign}+\mathbf{FSign})(\Sigma)}: s' \to T'\>$ is}\\
\mbox{the pushout for $\<id_{\mathbf{RSign}(\Sigma)}:\<\mathbf{RSign}(\Sigma), \emptyset\> \to I,  \mathit{in}_{l}: \<\mathbf{RSign}(\Sigma), \emptyset\> \to s'\>$} \\
\mbox{then: }
\mbox{for all $\beta \in \mathbf{Sen}^{\sf Eq}(\mathbf{RSign}(\Sigma)+(\mathbf{FSign}(\Sigma)/\{x\}))$, }\\
\mbox{$T \vdash^{\sf Eq}_{(\mathbf{RSign}+\mathbf{FSign})(\Sigma)} \mathbf{Sen}^{\sf Eq}(\varphi)(\beta)$ iff $T' \vdash^{\sf Eq}_{(\mathbf{RSign}+\mathbf{FSign})(\Sigma)} \mathbf{Sen}^{\sf Eq}(\varphi)(\beta)$}\\
\mbox{such that $\varphi$ is the identity over the signature $\mathbf{RSign}(\Sigma)+(\mathbf{FSign}(\Sigma)/\{x\})$, }\\
\mbox{mapped to the signature $\mathbf{RSign}(\Sigma)+(\mathbf{FSign}(\Sigma)$, and $\<I, \mathcal{M}, in_r (s')\> \models^{\sf FOCTL*}_\Sigma \alpha$ }\\
\<I, \mathcal{M}, in_l (s)\> \models^{\sf FOCTL*}_\Sigma \eee \varphi \mbox{ iff } \mbox{there exists $\pi \in |({\it List}\ \(\mathbf{States} (\Sigma)|_{\mathcal{M}_{bs}}\)|$, such that:}\\
\hfill \begin{array}{l} 
\pi[0] = s \mbox{ and for all $i \in \NAT$, $\pi[i]\ T^\mathcal{M}\ \pi[i+1]$, and}\\
\<I, \mathcal{M}, in_r (\pi)\> \models^{\sf FOCTL*}_\Sigma \varphi \end{array}
\end{array}
\]
\[
\begin{array}{l}
\<I, \mathcal{M}, in_r (\pi)\> \models^{\sf FOCTL*}_\Sigma \alpha \mbox{ iff } \<I, \mathcal{M}, in_l (\pi[0])\> \models^{\sf FOCTL*}_\Sigma \alpha\\
\hfill \mbox{, for all $\alpha \in \mathit{FormStFOCTL*} (\<C, F, P\>)$}\\
\<I, \mathcal{M}, in_r (\pi)\> \models^{\sf FOCTL*}_\Sigma \neg\varphi \mbox{ iff }  \<I, \mathcal{M}, in_r (\pi)\> \models^{\sf FOCTL*}_\Sigma \varphi \mbox{ does not hold}\\
\<I, \mathcal{M}, in_r (\pi)\> \models^{\sf FOCTL*}_\Sigma \varphi \lor \psi \mbox{ iff } \\
\hfill \<I, \mathcal{M}, in_r (\pi)\> \models^{\sf FOCTL*}_\Sigma \varphi$ or $\<I, \mathcal{M}, in_r (\pi)\> \models^{\sf FOCTL*}_\Sigma \psi \\
\<I, \mathcal{M}, in_r (\pi)\> \models^{\sf FOCTL*}_\Sigma \X \varphi \mbox{ iff } \<I, \mathcal{M}, in_l (\pi[1])\>  \models^{\sf FOCTL*}_\Sigma \varphi \\
\<I, \mathcal{M}, in_r (\pi)\> \models^{\sf FOCTL*}_\Sigma \varphi \U \psi \mbox{ iff } \mbox{there exists $i \in \NAT$ such that:}\\
\hfill \begin{array}{l} 
\<I, \mathcal{M}, in_r (\pi[i..])\> \models^{\sf FOCTL*}_\Sigma \psi \\
\mbox{for all $j \in \NAT$, $j < i$ implies $\<I, \mathcal{M}, in_r (\pi[j..])\> \models^{\sf FOCTL*}_\Sigma \varphi$}\end{array}
\end{array}
\]

Once again, the reader should note that both $\pi[i]$ and $\pi[i..]$ must be defined within the formal framework chosen for defining ${\sf Struct}^{\sf FOCTL*}$ as we did in Exs.~\ref{ex:ltl} and~\ref{ex:ctl}.

\section{Final remarks and Conclusions}
\label{conclusions}
In the present work, we discussed the role of model theory in software design and analysis through logical reasoning. More precisely, we focussed on the inappropriateness of using abstract model theory, a conception where models are unstructured points in a class or, in the best case, a collection of naively defined elements, as such a view conceptualise domains of interpretation as a purely abstract set, disregarding the fact that in the context of software specification we only care about values that can be obtained by the application of the functions declared in the signature of the available modules/components/types/etc.

In contrast to abstract model theory, we rely on the idea, borrowed from initial semantics \cite{ehrig85,goguen:swat74}, of formalising values as terms and, from there, we moved on to formalising interpretations (of the rigid symbols) and states of a system (assigning semantics to flexible ones) as equational theories, prescribing what hold, as what can be proved in the categorical glueing, through a pushout, of the interpretation and the state. 

Logical languages with relational semantics, like most modal and hybrid logics, are ubiquitous in software specification as they generally expose specific, and generally dynamic, properties of software artefacts. Relational models, known under the generic name of Kripke structures, provide the support for understanding the transitions between the states of an evolving software system. Even when most of such languages share the motivation of reflecting certain dynamic behaviour of software artefacts, the type of properties they characterise are of different nature, thus requiring some aspects of the models to be tailored to its specific purpose (for example, linear ordering of states of the relational structure for linear logics, like ${\sf LTL}$, states for pure branching time, like ${\sf CTL}$, and dynamic logics, like ${\sf PDL}$, a mixture of both for non-pure branching time variants, like ${\sf CTL}^*$, etc.).

Therefore, the main result of this paper is the definition, within the field of Institutions, of a unified framework for describing classes of relational models, supported by a sound and complete calculus to reason about them, and whose states provide a concrete representation of values as terms over an appropriate signature. To accomplish that, we split models into their static elements, further classified into rigid and flexible, represented by \emph{interpretations} and \emph{states}, respectively and the dynamic elements, understood as labelled transition systems determined by binary relations, constrained by means of a relational theory. The static elements are completely axiomatised in (higher order) equational logic, while the dynamic elements are characterised by the complete calculus of the elementary theory of binary relations with closure.

While the examples presented in previous sections expose the versatility of the framework, we identify two interesting directions of further investigation. On the one hand, we identify a limitation of the framework in the fact that it cannot support the semantics of logical languages whose transition relations are flexible (ours are interpreted as rigid). Supporting such a semantics requires the extension of the notion of state so it can assign meaning to flexible relational symbols. On the other hand, while the formalisation of the static elements (both the notion of interpretation and state) avoids any (explicit) use of semantics, the relational aspects are represented by the explicit use of the models of the relational theory. Formalising the relational aspects in purely proof theoretical terms, as we did for the static elements of relational models, requires the capability of associating syntactic terms (denoting individuals) of ${\sf ETBR*}$ to states in order to force a unique interpretation, providing the means for expressing atomic relational formulae (like $s_0\ T\ s_1$) so relational structures can be characterised in axiomatic terms.

\bibliography{bibdatabase}
\bibliographystyle{splncs}

\end{document}